\documentclass[11pt]{article}
\usepackage[margin=1in]{geometry}
\usepackage{amsfonts,amsmath,latexsym,amssymb,amsthm}
\usepackage[USenglish]{babel}
\usepackage{microtype}

\usepackage{graphicx}
\usepackage[dvipsnames]{xcolor}
\usepackage{xspace}
\usepackage{paralist}
\usepackage{enumitem}
\usepackage{bm}
\usepackage{cancel}
\usepackage{authblk}
\usepackage{thm-restate}

\usepackage[utf8]{inputenc}
\usepackage{lmodern}

\usepackage{tikz}
\usetikzlibrary{patterns}
\usetikzlibrary{decorations.pathreplacing}
\usetikzlibrary{calligraphy}

\usepackage[algo2e,ruled,vlined]{algorithm2e}

\usepackage{hyperref}  
\hypersetup{pdftitle={Breaking the Barrier of 2 for the Competitiveness of Longest Queue Drop}}
\hypersetup{pdfauthor={Antonios Antoniadis, Matthias Englert, Nicolaos Matsakis, Pavel Vesel\'y}}

\title{Breaking the Barrier of~$2$\\ for the Competitiveness of Longest Queue Drop}
\date{} 

\author[1]{Antonios Antoniadis\thanks{Work done in part while the
		author was at Saarland University and Max-Planck-Institute for
		Informatics and supported by DFG grant AN 1262/1-1.}}
\author[2]{Matthias Englert}
\author[3]{Nicolaos Matsakis\thanks{Supported by GA \v{C}R project 22-22997S.}}
\author[4]{Pavel Vesel\'y\thanks{Work done in part while the author was at University of Warwick.
		Partially supported by European Research Council grant ERC-2014-CoG 647557, by GA \v{C}R project 22-22997S, and by Center for Foundations of Modern Computer Science (Charles Univ.\ project UNCE 24/SCI/008).}}

\affil[1]{University of Twente, The Netherlands, \texttt{a.antoniadis@utwente.nl}}
\affil[2]{University of Warwick, UK, \texttt{M.Englert@warwick.ac.uk}}
\affil[3]{Charles University, Czech Republic, \texttt{nickmatsakis@gmail.com}}
\affil[4]{Charles University, Czech Republic, \texttt{vesely@iuuk.mff.cuni.cz}}

\def\OPT{\textsf{OPT}\xspace}
\def\ALG{\textsf{ALG}\xspace}
\def\LQD{\textsf{LQD}\xspace}

\def\EXTRA{\textsf{OPT}_{\textsf{EXTRA}}\xspace}
\def\LEXTRA{\textsf{LQD}_{\textsf{EXTRA}}\xspace}

\DeclareMathOperator*{\argmin}{arg\,min}

\let\eps\varepsilon
\let\rho\varrho

\newcommand{\calA}{\mathcal{A}}

\newcommand{\calD}{\mathcal{D}}

\newcommand{\calL}{\mathcal{L}}
\newcommand{\calQ}{\mathcal{Q}}

\newcommand{\calU}{\mathcal{U}}

\newcounter{thm}
\AtEndPreamble{%
\newtheorem{theorem}[thm]{Theorem}

\newtheorem{lemma}[thm]{Lemma}
\newtheorem{fact}[thm]{Fact}
\newtheorem{claim}[thm]{Claim}

\newtheorem{definition}[thm]{Definition}
\newtheorem{observation}[thm]{Observation}
}

\newcounter{mynote}[section] \newcommand{\thenote}{\thesection.\arabic{mynote}}
\newcommand{\pvnote}[1]{\refstepcounter{mynote}{\color{blue}%
    \mathversion{bold}\marginpar{\hfill\tiny\sffamily\bfseries
      \textcolor{blue}{PV~\thenote}}$\ll$\bfseries\sffamily#1
    --Pavel$\gg$}\mathversion{normal}}
\newcommand{\menote}[1]{\refstepcounter{mynote}{\color{red}%
    \mathversion{bold}\marginpar{\hfill\tiny\sffamily\bfseries
      \textcolor{red}{ME~\thenote}}$\ll$\bfseries\sffamily#1
    --Matthias$\gg$}\mathversion{normal}}
\newcommand{\aanote}[1]{\refstepcounter{mynote}{\color{Blue}%
    \mathversion{bold}\marginpar{\hfill\tiny\sffamily\bfseries
      \textcolor{Red}{AA~\thenote}}$\ll$\bfseries\sffamily#1
    --Antonios$\gg$}\mathversion{normal}}
\newcommand{\nmnote}[1]{\refstepcounter{mynote}{\color{Orange}%
    \mathversion{bold}\marginpar{\hfill\tiny\sffamily\bfseries
      \textcolor{Orange}{NM~\thenote}}$\ll$\bfseries\sffamily#1
    --Nick$\gg$}\mathversion{normal}}

\newcommand{\pvchanged}[1]{{\color{blue}#1}}

\renewcommand{\pvchanged}[1]{#1}
\renewcommand{\pvnote}[1]{}
\renewcommand{\nmnote}[1]{}
\renewcommand{\aanote}[1]{}
\renewcommand{\menote}[1]{}

\begin{document}
\maketitle

\begin{abstract}
  We consider the problem of managing the buffer of a shared-memory switch that transmits
  packets of unit value. A shared-memory switch consists
  of an input port, a number of output ports, and a buffer with a specific capacity. In each time step, an arbitrary
  number of packets arrive at the input port, each packet designated for one
  output port. Each packet is added to the queue
  of the respective output port.
  If the total number of packets exceeds the capacity of the buffer,
  some packets have to be irrevocably evicted. At the end of each time step, each output port
  transmits a packet in its queue and the goal is to maximize
  the number of transmitted packets.

  The Longest Queue Drop (\LQD) online algorithm accepts any arriving
  packet to the buffer. However, if this results in the buffer exceeding its memory
  capacity, then \LQD drops  a packet from whichever queue
  is currently the longest, breaking ties arbitrarily.
  The \LQD   algorithm was first introduced in 1991, and is known
  to be $2$-competitive since 2001. %
  Although \LQD remains the best known online algorithm for the problem and is of practical interest, determining
  its true competitiveness is a long-standing open problem. We show that \LQD is 1.6918-competitive, establishing the first
  $(2-\eps)$ upper bound for the competitive ratio of \LQD, for a
  constant $\eps>0$.
\end{abstract}

\section{Introduction}
The fact that communication networks are omnipresent highlights
the significance of improving their performance. A
natural way to achieve such performance improvements is to
develop better algorithms for buffer management of shared-memory
switches which form the lower levels of network communication.
We study a fundamental model of such switches.

Consider a shared-memory network switch consisting of a
buffer of size $M\in\mathbb{N}$, an
input port, and $N\in\mathbb{N}$ output ports.
 Furthermore, consider a slotted time model. In
each time step, an arbitrary number of unit-valued packets arrives to the input
port. Each packet comes with a label specifying the output port that it has to be forwarded
to. A buffer management algorithm has to make a decision for each
packet: either irrevocably evict it, or accept it while ensuring
that the buffer capacity $M$ is respected, which may mean that a previously accepted packet has to be evicted. At the end of the time step, each
output port with at least one packet in the buffer destined to it transmits a packet.
The goal of the
buffer management algorithm is to accept/evict incoming packets or evict already accepted packets, so
as to maximize the throughput, i.e., the total number of transmitted packets, while
ensuring that at most $M$ packets in total are stored for all output ports at any time.

Given the inherently online nature of buffer management problems,
a standard approach is to design online algorithms for them and evaluate the algorithm's performance
using its competitive ratio. More specifically, an online algorithm \ALG is $c$-competitive (where $c\ge 1$)
if the number of packets transmitted by an optimal offline algorithm \OPT (that has full knowledge of the incoming packet sequence a priori) 
is at most $c$ times the number of packets transmitted by \ALG. %
There exists an extensive body of research dedicated to
designing competitive online algorithms with the aim of improving the
performance of networking devices that incorporate buffers (see e.g.~\cite{Goldwasser10,NikolenkoK16}).

Since packets have unit value, we can assume without loss of generality that the packets destined to a
specific output port are transmitted in an earliest-arrival  (FIFO)
fashion and thus, it is helpful to associate each output port with
a queue.

Intuitively speaking, to maximize throughput, one would like to maintain a flow of packet transmissions for as many queues in parallel as possible. It is therefore
desirable to prioritize accepting packets for queues
that do not have many incoming packets in the near
future. Unfortunately, an online algorithm does not know which queues
these are, and in order to be insured against an adversarial input it
seems reasonable to try to keep the queue lengths as balanced as
possible in every step. This is exactly the idea behind the
online algorithm \emph{Longest Queue Drop (\LQD)},
introduced in 1991 by Wei, Coyle, and Hsiao~\cite{WCH}:  The incoming packet is always accepted and if this causes the buffer to exceed its capacity
then one packet from the longest queue, breaking ties arbitrarily, is dropped, i.e., evicted (this could be the incoming packet).\footnote{
Wei, Coyle, and Hsiao proposed the \LQD algorithm for the problem of shared-memory switches, consisting of $N$ input ports and $N$ output ports.
Rather than assuming $N$ input ports, each of which may receive at most one packet per time step, we more generally assume that there is a single input port of infinite capacity. Furthermore, we do not put any restrictions on the number of output ports, i.e., we allow $N$ to be arbitrarily large. Again, this only makes the problem more general.}

The \LQD algorithm, apart from being a natural online algorithm to
derive, remains the only known competitive algorithm for this problem. 
Since the algorithm is simple and can be used, for instance, to achieve a fair distribution of the bandwidth,
it is of some practical interest; see e.g.~\cite{bruno1998early, chamberland2000overall, ChaoGuoBook, ChaoLiuBook, nabeshima2005performance, suter1998design}.

\paragraph*{Previous results}
Hahne, Kesselman, and Mansour~\cite{HahneKM01} provided the first formal analysis of \LQD,
showing that it is 2-competitive (see also Aiello, Kesselman, and Mansour~\cite{AielloKM08}). The proof follows from a simple procedure that charges the extra profit of \OPT to the profit of \LQD.
Furthermore, they demonstrated that \LQD is at least $\sqrt{2}$-competitive, and also showed a general lower bound of $4/3$
for the competitive ratio of any deterministic online algorithm.

The analysis of \LQD in~\cite{AielloKM08,HahneKM01} was then refined by Kobayashi, Miyazaki, and Okabe~\cite{KobayashiMO07}
who showed that the \LQD competitive ratio is at most $2 - \min_{k=1,\dots, N} (\lfloor M/k\rfloor + k - 1) / M$.
However, for $N > \sqrt{M}$, this bound becomes 
$2 - O(1/\sqrt{M})$ and therefore does not establish a $2-\eps$ upper bound for a constant $\eps>0$ in general. Additionally, for the case of $N=2$ output ports, Kobayashi et al.~\cite{KobayashiMO07} claimed that \LQD is exactly $\frac{4M-4}{3M-2}$-competitive (we note that although this result holds for an even buffer size, the argument unfortunately breaks down when the buffer size is odd).
For the case of $N=3$ output ports, Matsakis showed that \LQD is 1.5-competitive \cite{Matsakis15}.

More recently, Bochkov, Davydow, Gaevoy, and Nikolenko~\cite{bochkov2019new} improved the lower bound on the competitiveness of \LQD
from $\sqrt{2}$ to approximately $1.44$ (using a direct simulation of \LQD and also independently, by solving a linear program).
Moreover, they showed that any deterministic online algorithm is at least $\sqrt{2}$-competitive, using a construction inspired by the \LQD specific 
lower bound from~\cite{AielloKM08,HahneKM01}.
To the best of our knowledge, so far, no randomized algorithms for this problem have been studied.

\paragraph*{Our contribution}
  Although \LQD is %
  the best known online algorithm for buffer management in shared-memory switches, determining
  its true competitiveness remains an elusive problem and has
  been described as a significant open problem in buffer
  management~\cite{Goldwasser10,NikolenkoK16}.
After the initial analysis which showed that \LQD is 2-competitive and not better than $\sqrt{2}$-competitive~\cite{AielloKM08,HahneKM01}
progress on the upper bound has been limited to special cases (e.g., with restrictions on the number of output ports or memory size) \cite{KobayashiMO07,Matsakis15}.
In this paper, we make the first major progress since 2001 on upper bounding the competitive ratio
of \LQD. Namely, we prove the first $(2-\eps)$ upper bound for a constant $\eps>0$ without restrictions on the number of ports or the size of the buffer:
\begin{theorem}
  \label{thm:main}
  \LQD is $1.6918$-competitive.
\end{theorem}

We remark that Theorem~\ref{thm:main} applies to \LQD with any
tie-breaking rule, even if tie-breaking
is under control of the adversary.

\paragraph*{Our techniques}
The proof of 2-competitiveness of \LQD in~\cite{AielloKM08,HahneKM01} uses the following general approach.
If an optimal offline algorithm \OPT currently stores more packets for a queue than \LQD does, these excess packets present potential extra profit for \OPT.
Each such potential extra packet $p$ in \OPT is then matched to a packet that is transmitted by \LQD at some point before packet $p$ can be transmitted by \OPT.

Our approach is different in that we (for the most part) do not match specific packets to one another. Instead, the idea is to take the total profit
of \LQD in each step and distribute it evenly among all potential extra packets that exist at the time. As such, the scheme is less discrete than the previous one.
We then carefully calculate that, for each queue, \emph{on average} each potential extra packet in that queue receives a profit strictly larger than one.

As described here, this approach does not quite work yet. Two additional types of charging concepts have to be combined with this first idea: One involves not splitting the \LQD profit completely evenly and instead slightly favoring queues with relatively few potential extra packets, and the other involves matching some of the potential extra packets of \OPT to extra packets that \LQD transmits.
Another difficulty is that the lengths of two queues from which packets are evicted in the same time step may differ by one packet. This makes our proof more intricate. To deal with this, we introduce a potential function that will amortize the \LQD profit in a suitable way. %
Then, the main challenge is to obtain useful lower bounds on the profit assigned to each queue, for which we introduce a novel scheme that relates the buffers of \LQD and of \OPT.

\paragraph*{Further related work}
We refer the reader to the survey by Goldwasser~\cite{Goldwasser10} for an overview of online algorithms
for buffer management problems. Additionally, the survey of Nikolenko and Kogan~\cite{NikolenkoK16} incorporates some more recent work. 
In the following, we discuss some of the results related to online buffer management for switches.
In general, buffer management algorithms can be partitioned into
\emph{preemptive} ones, i.e., algorithms that allow for the eviction
of already accepted packets from the buffer (eviction is also referred to as preemption), and
\emph{non-preemptive} ones that never evict a packet after it has been accepted.

Kesselman and Mansour~\cite{KesselmanM04} study buffer management in shared-memory switches in the non-preemptive setting in which a packet has to be transmitted once it is stored in the buffer and can no longer be evicted. They introduce the Harmonic online algorithm, which tries to maintain the length of the $i^{th}$ longest queue as roughly proportional to a $1/i$ fraction of the memory. They show that this algorithm is $(\ln(N)+2)$-competitive and give a general lower bound of $\Omega(\log N/\log \log N)$ for the performance of any deterministic non-preemptive online algorithm. Considering the non-constant lower bound that they establish, it follows that preemption provides a significant advantage.

Eugster, Kogan, Nikolenko, and Sirotkin~\cite{EKNS14} generalize the same problem in the following two ways: First, they study unit-valued packets labeled with an output port and a processing requirement (in our case, we have a unit processing cycle per packet). Packets accepted to the same queue have the same processing requirement. They introduce the preemptive Longest-Work-Drop algorithm: If the buffer is not full, the incoming packet is accepted; otherwise, a packet is preempted from a queue that has the largest total processing requirement. They show that this algorithm is 2-competitive and at least $\sqrt{2}$-competitive and that the competitive ratio of \LQD for this more general problem is at least $(\sqrt{k}-o(\sqrt{k}))$, where $k$ is the maximum processing time of any packet. Second, they address the problem of different packet values when all packets have unit processing requirements. It is proven that \LQD is at least $(\sqrt[3]{K}-o(\sqrt[3]K))$-competitive in this case, where $K$ is the maximum packet value. They also introduce a new algorithm which they conjecture to have a constant competitive ratio. %

Azar and Richter~\cite{AzarR05} study switches with multiple input
queues. More specifically, they consider one output port and $N$ input
ports and assume that each input port has an independent buffer of
size $M$. At each time step, one packet can be sent from a single
input port to the output port. For $M=1$, they prove a lower bound of
$1.46-\Theta(1/N)$ for the competitive ratio of any randomized
online algorithm and a lower bound of $2-1/N$ for deterministic online algorithms. They also give a randomized $\frac{e}{e-1}\approx 1.582$-competitive algorithm for $M>\log N$.
For $M > 1$, Albers and Schmidt~\cite{AlbersS05}
design a deterministic 1.889-competitive algorithm for this problem and show a
deterministic lower bound of $\frac{e}{e-1}\approx 1.582$ when $N\gg M$. 
Azar and Litichevskey~\cite{AzarL06} give a deterministic online algorithm matching this bound for large $M$.

A lot of research has been dedicated to the natural single input and single output port model.
The model is trivial for unit packet values, but challenging if packets can have different values and the goal is to maximize the total value of transmitted packets.
There exists a single queue for the accepted packets and one of the most studied versions of this problem requires packet transmission in the FIFO order.
Kesselman, Lotker, Mansour, Patt-Shamir, Schieber, and Sviridenko~\cite{KesselmanLMPSS04} show that a simple greedy algorithm is exactly $(2-1/(M+1))$-competitive when preemption is allowed.
A series of works gradually improved the analysis of a better online algorithm from 1.983~\cite{KesselmanMS05}, over 7/4~\cite{BFKMSS04}, to $\sqrt{3}$~\cite{EnglertW09}. Kesselman, Mansour, and van Stee~\cite{KesselmanMS05} also show a general lower bound of 1.419 for the competitive ratio of any preemptive deterministic online algorithm. 

The authors of~\cite{KesselmanLMPSS04} introduce the bounded-delay model of single output port switches.
In this model, the buffer has unlimited size and allows for packets to be transmitted in any order, however, each packet has
a deadline after which it needs to be dropped from the buffer. Once again, the problem is only interesting if packets can have different values.
Any deterministic online algorithm is at least $\phi\approx 1.618$-competitive~\cite{andelman_queueing_policies_03,chin_partial_job_values_03,hajek_unit_packets_01,Zhu04},
and after a sequence of gradual improvements~\cite{CJST07,EW12,LSS07},  Vesel\'y, Chrobak,  Je\.{z}, and Sgall~\cite{VeselyCJS22} gave a $\phi$-competitive algorithm.
The competitive ratio of randomized algorithms is still open, with the best upper bound of
$\frac{e}{e-1}\approx 1.582$~\cite{bienkowski_randomized_algorithms_11,CCFJST06,Jez13}
(that holds even against the adaptive adversary), while the lower bounds are $1.25$ against the oblivious adversary~\cite{bienkowski_randomized_algorithms_11}
and $4/3$ against the adaptive adversary~\cite{chin_partial_job_values_03}.

Lastly, we mention the model of Combined Input and Output Queued (CIOQ) Switches, in which the switch has $N$ input ports and $N$ output ports. Each input and output port has its own buffer 
and each input port can transfer a packet to any output port; however, at most one packet can be sent from any input port and at most one packet can be accepted by any output port, during one transfer cycle of the switch. A parameter $S$ called \emph{speedup} equals the number of transfer cycles of the switch taking place per one time step. For the unit-value case, Kesselman and Ros\'{e}n \cite{KesselmanR06} provide a 2-competitive non-preemptive online algorithm for $S=1$, which becomes 3-competitive for any $S$. A faster algorithm with the same competitive ratio is given by Al-Bawani, Englert, and Westermann~\cite{AEW18a}.

\section{Setup of the Analysis}
\label{s:setup of analysis}

We fix an arbitrary instance $I$. Let \OPT and \LQD be the optimal offline algorithm and the Longest Queue Drop algorithm, respectively. In a slight abuse of notation, we also denote the profit that the optimal offline algorithm gains on input instance $I$ as \OPT and the profit that the Longest Queue Drop algorithm gains as \LQD.
Our goal is to give an upper bound on \OPT/\LQD.
\pvchanged{More precisely, we will show $\OPT\le c\cdot \LQD + O(M)$ for $c\approx 1.6918$.
	As usual for such problems,
	a competitive analysis with a bounded additive error (more precisely, sublinear in the profit of \OPT)
	yields an upper bound on the strict competitive ratio, without any additive term. This follows since any instance $I$ may be repeated arbitrarily many times (starting the next copy of $I$ once buffers of both \OPT and \LQD are empty and no further packet from the previous copy will arrive).
	Thus the additive error can be made negligible by taking sufficiently many copies of the instance.}

For a time step $t$ and a queue $q$, we say that \OPT transmits an \OPT-extra packet from $q$ if \OPT transmits a packet from $q$ in step $t$ but \LQD does not. Equivalently, queue $q$ is non-empty in \OPT's buffer but empty in \LQD's buffer at $t$.
Similarly, we say that \LQD transmits an \LQD-extra packet from a queue $q$ in step $t$ if \LQD transmits a packet from $q$ at $t$ but \OPT does not. 

Let $\EXTRA$ and $\LEXTRA$ be the total number of transmitted \OPT-extra and \LQD-extra packets, respectively, over all time steps and queues. Then
$\OPT-\EXTRA=\LQD-\LEXTRA$ and hence
$\frac{\OPT}{\LQD} = 1+ \frac{\EXTRA - \LEXTRA}{\LQD}$. 
Therefore, if we show that 
$\displaystyle \rho\cdot (\EXTRA - \LEXTRA) \le \LQD + O(M)$
for some $\rho > 1$, it will imply a competitive ratio of $1 + 1 / \rho < 2$ for the Longest Queue Drop algorithm.

Let $e_q$ denote the total number of transmitted \OPT-extra packets from queue $q$ over all time steps. Then we have $\EXTRA = \sum_q e_q$ and we will show
\begin{equation}\label{eqn:mainIneq}
\rho\cdot \left(\sum_q e_q - \LEXTRA\right) \le \LQD + O(M) \enspace.
\end{equation}
We now give a high-level overview of the proof of Equation~\eqref{eqn:mainIneq}, which consists of
two parts: (i) splitting the \LQD profit among queues $q$ with $e_q>0$, and
(ii) mapping transmitted \LQD-extra packets to queues $q$ with $e_q>0$.

For~(ii), we use the term $\LEXTRA$ in~\eqref{eqn:mainIneq} to ``cancel out''
some transmitted \OPT-extra packets.
To this end, we will define how each transmitted \LQD-extra packet $p$ is mapped to a queue $q$ (which is different from the one $p$ is transmitted from). Let $m_q$ be the number of transmitted \LQD-extra packets which are mapped to $q$. The mapping will be such that  $\sum_q m_q \le \LEXTRA$ and that $m_q \le e_q$. Define $\hat{e}_q = e_q-m_q \ge0$ as the number of \OPT-extra packets transmitted from queue $q$ which are not canceled out. %

We have $\sum_q e_q - \LEXTRA \le \sum_{q} (e_q - m_q) = \sum_{q} \hat{e}_q$.
Hence, it is sufficient for each $q$ to receive an \LQD profit of at least $\rho\cdot \hat{e}_q$,
from which it follows that $\rho\cdot \sum_{q}\hat{e}_q \le \LQD$,
implying~\eqref{eqn:mainIneq}.  

We describe splitting the \LQD profit, enhanced with a suitable potential, in Section~\ref{s:splitting} and introduce useful quantities for bounding the profit
assigned to a particular queue in Section~\ref{s:liveAndDying}.
Then, in Section~\ref{sec:mapping}, we introduce the mapping of transmitted \LQD-extra packets to queues
and derive a relation between the buffers of \LQD and of \OPT.
Finally, we put the bounds together and optimize the value of $\rho$ in Section~\ref{sec:puttingItAllTogether},
which will yield our upper bound on the \LQD competitive ratio.
Table~\ref{tab:glossary} in Appendix~\ref{sec:notations} provides a list of concepts and notation used throughout the paper, most of which are defined in the subsequent section.

\section{Splitting the \LQD Profit}
\label{s:splitting}

In this section, we explain how the \LQD profit is split.
Before we proceed, %
we introduce some notation and terminology and define a key time step for a queue.
We index the time steps by integers starting from $0$.
When we refer to the state of a queue at time step $t$ under some
algorithm, we refer to the state after all new packets of step $t$
have arrived and after all possible evictions of packets by the algorithm, but
before any packet is transmitted by the algorithm at the end of step $t$.
We use the following notation:
 
\begin{description}%
\item{$s^t_{\OPT}(q)$:} the number of packets in queue $q$ in the \OPT buffer in step $t$,
\item{$s^t_{\LQD}(q)$:} the number of packets in queue $q$ in the \LQD{} buffer in step $t$,
\item{$s^t_{\max} = \max_q s^t_{\LQD}(q)$:} the maximal size of a
  queue in the \LQD buffer in step $t$. 
\end{description}

We say that a queue $q$ is \emph{active} in a time step $t$ if $s^{t}_{\OPT}(q)\ge 1$ or $s^{t}_{\LQD}(q)\ge 1$. Otherwise, if $q$ is empty in both buffers at $t$, we say that $q$ is \emph{inactive} at $t$.
See Figure~\ref{fig:buffer} for an illustration.

\paragraph*{Assumptions on the instance}
We make the following assumption w.l.o.g., which will greatly reduce the additional notation required.
\begin{description}%
\item[(A1)] For any queue $q$ and step $t$, we assume that if $s^t_{\LQD}(q) \le 1$ %
but at least one packet arrived to $q$ at or before time step $t$, then no packet arrives to queue $q$ after step $t$.
(If $s^t_{\LQD}(q) = 1$ then the last packet is transmitted from $q$ in step $t$.)
\end{description}

To see that this assumption is w.l.o.g., we iteratively modify the instance under consideration as follows:
Let $q$ be any queue that does not satisfy this assumption and
let $t$ be the first time step such that $s^t_{\LQD}(q) \le 1$ and there is a packet arriving to $q$ at $t$ or before.
\pvchanged{
	As $q$ does not satisfy the assumption, there is a packet arriving to $q$ after step $t$. %
In the modified instance, all such arriving packets for queue $q$ are instead sent to a new queue $q'$ which is not used in the instance otherwise. Observe that the profit of \LQD does not change after redirecting these packets to a new queue.
On the other hand, the profit of \OPT cannot decrease when we make this change, since any packet arriving to the new queue $q'$
would be stored by \OPT if it is stored before this change
(\OPT could potentially save some capacity in its buffer due to this change as it may sent packets from both $q$ and $q'$).
}
We remark that the new queue is always available as the number of output ports $N$ is not restricted and can be arbitrarily large.
Note that we only make this assumption to simplify our notation and it does not affect the generality of our analysis. Indeed, if the number of output ports used in the original instance is bounded by $N_0$,
then after applying this transformation, there are always at most $N_0$ queues non-empty for $\LQD$ at any one time.

For instance, under assumption~(A1), %
if an \OPT-extra packet is transmitted from a queue $q$ in some step $t$ %
(when $q$ is empty for \LQD), then no packet arrives to $q$ in any step $t' > t$. %

\begin{figure}[t]
\begin{centering}
\begin{tikzpicture}
 \draw[color=blue, pattern = north west lines, pattern color = blue ] (0,0)
  -- (0,6) -- (0.45,6) -- (0.45,0) -- (1,0) -- (1,5.5) -- (1.45,5.5) --
  (1.45,0) -- (2,0) -- (2,5.5) -- (2.45, 5.5) -- (2.45,0) -- (3,0) --
  (3,3.5) -- (3.45,3.5) -- (3.45,0) -- (4,0) -- (4,3.5) -- (4.45,3.5) --
  (4.45,0) -- (5,0) -- (5,3.5) -- (5.45,3.5) -- (5.45,0) -- (6,0) -- (6,3)
  -- (6.45,3) --(6.45,0) -- (7,0) -- (7,2) -- (7.45,2) -- (7.45,0);
  \draw[color = red, pattern = north east lines, pattern color = red]
  (0.55,0) -- (0.55, 4) -- (1,4) -- (1,0) -- (1.55,0)-- (1.55,7) -- (2,7)
  -- (2,0) -- (2.55,0) -- (2.55,2.5) -- (3,2.5)
  -- (3,0) -- (3.55,0) -- (3.55,0.5) -- (4,0.5) -- (4,0) -- (4.55,0) --
  (4.55,0.5) -- (5,0.5) -- (5,0) -- (5.55,0) -- (5.55,3.5) -- (5.55,7) --
  (6,7) -- (6,0) -- (6.55,0) --
  (6.55,6) -- (7,6) -- (7,0) -- (7.55,0) -- (7.55,1.5)
  -- (8,1.5) --
  (8,0) -- (8.55,0) -- (8.55,3.5) -- (9,3.5) -- (9,0);
  \draw[very thick, ->] (0,0) -- (11,0);
  \draw[very thick, ->] (0,0) -- (0,7.5);
  \node () at (0.5,-0.5) {$\downarrow$ $1$};
  \node () at (1.5,-0.5) {$\downarrow$ $2$};
  \node () at (2.5,-0.5) {$\downarrow$ $3$};
  \node () at (3.5,-0.5) {$\downarrow$ $4$}; 
  \node () at (4.5,-0.5) {$\downarrow$ $5$};
  \node () at (5.5,-0.5) {$\downarrow$ $6$};
  \node () at (6.5,-0.5) {$\downarrow$ $7$};
  \node () at (7.5,-0.5) {$\downarrow$ $8$};
  \node () at (8.5,-0.5) {$\downarrow$ $9$};
  \node () at (9.5,-0.5) {$\downarrow$ $10$};
  \node () at (11,-0.5) {queues};
  \draw[very thick, dashed] (1,-0.5) -- (1,7.5);
  \draw[very thick, dashed] (2,-0.5) -- (2,7.5);
  \draw[very thick, dashed] (3,-0.5) -- (3,7.5);
  \draw[very thick, dashed] (4,-0.5) -- (4,7.5);
  \draw[very thick, dashed] (5,-0.5) -- (5,7.5);  
  \draw[very thick, dashed] (6,-0.5) -- (6,7.5);
  \draw[very thick, dashed] (7,-0.5) -- (7,7.5);
  \draw[very thick, dashed] (8,-0.5) -- (8,7.5);
  \draw[very thick, dashed] (9,-0.5) -- (9,7.5);
  \draw[very thick, dashed] (10,-0.5) -- (10,7.5);
  \node () at (-0.8,7.2) {packets};
  \node () at (-0.4,0.25) {$1$};
  \node () at (-0.4,0.75) {$2$};
  \node () at (-0.4,1.25) {$3$};
  \node () at (-0.4,1.75) {$4$};
  \node () at (-0.4,2.25) {$5$};
  \node () at (-0.4,2.75) {$6$};
  \node () at (-0.4,3.25) {$7$};
  \node () at (-0.4,3.75) {$8$};
  \node () at (-0.4,4.25) {$9$};
  \node () at (-0.4,4.75) {$10$};
  \node () at (-0.4,5.25) {$11$};
  \node () at (-0.4,5.75) {$12$};
  \node () at (-0.4,6.25) {$13$};
  \node () at (-0.4,6.75) {$14$};
 \end{tikzpicture}
\caption{An example of the buffer configuration for \LQD and \OPT during some time step
  $t$ while processing the incoming packets, for buffer of size $M=65$. The blue, north-west shaded areas (aligned to the left) correspond to the packets in queues of
  \LQD and the red, north-east shaded areas (aligned to the right) to the queues of \OPT.
For instance, we
  have $s_{\OPT}^t(6)= 14$, and $s_{\LQD}^t(6) = 7$. Furthermore, $s_{\text{max}}^t = 12$
  is the maximal size of a queue for \LQD.
  Note that an \OPT-extra packet is going to be transmitted from queue~$9$ in step $t$,
  and as queue~$9$ is empty for \LQD, no further packet arrives to this queue by assumption~(A1). %
  All the queues with an index $\ge 10$ are inactive (i.e., empty in both buffers).
  According to Definition~\ref{def:overflow}, queues $1, 2,$ and $3$ overflow.
  As an example, assume that further $3$ packets arrive into queue $8$ and the \LQD buffer is already full. Then \LQD
  would first evict a packet from queue $1$ and then select two of the
  queues $1$, $2$ or $3$, dropping one packet from each selected queue.
  }
\label{fig:buffer}
\end{centering}
\end{figure}

\paragraph*{Overflowing queues}
Intuitively, if a packet destined to $q$ is evicted by \LQD at $t$,
then we say that $q$ \emph{overflows}. Furthermore, in such a case, the \LQD buffer is full in step $t$ and
$q$ has $s^t_{\max} - 1$ or $s^{t}_{\max}$ packets at $t$ (see the example in Figure~\ref{fig:buffer}).
This possible difference of 1 in the lengths of two different overflowing queues makes our analysis substantially more involved.\footnote{A less sophisticated version of our proof, which deals with this scenario in a less careful way, only gives an upper bound of about 1.906 on the competitive ratio.
Nevertheless, this analysis still requires the majority of concepts, lemmas, and calculations developed in the paper.}
For technical reasons, we also call a queue $q'$ containing at least $s^t_{\max} - 1$ packets at $t$
overflowing, provided that the \LQD buffer is full and $s^t_{\LQD}(q') \ge 1$, even though there may be no packet for $q'$ that is
evicted at time $t$.

\begin{definition}\label{def:overflow}
We say that a queue $q$ \emph{overflows} in step $t$ if 
\begin{enumerate}[label=(\roman*)]
\item  a packet destined to $q$ is evicted by \LQD at $t$, or
\item  the \LQD buffer is full in step $t$ and $s^t_{\LQD}(q)\ge \max(s^t_{\max} - 1, 1)$,
\end{enumerate}
or both.
\end{definition}

Assumption~(A1) %
implies that once $s^t_{\LQD}(q) \le 1$ but at least one packet arrived to $q$ at or before step $t$, then 
queue $q$ does not overflow after $t$ (as after step $t$, no packet arrives to $q$ and $q$ thus remains empty in the \LQD buffer). 
		We remark that, somewhat counterintuitively, it may even happen that no packet destined to any queue gets evicted at $t$
		but there are still some overflowing queues, provided that the \LQD buffer is full.
		Furthermore, an empty queue $q$ may overflow in some step $t$ if there is a packet destined to $q$ that is evicted at $t$, however,
		as noted above, this queue will not overflow in any step after $t$ (also note that when an empty queue overflows, all queues of \LQD contain at most a single packet).

\paragraph*{Key time step}
Based on assumption~(A1), %
we define the key time step $t_q$ for each queue $q$:
\begin{description}
\item{$t_q$:} the last time step in which queue $q$ overflows; if $q$ does not overflow in any step, we define $t_q = -1$ (recall that we index time steps starting from $0$).
\end{description}
Some important properties follow directly from the definition of $t_q$:
No packet is ever evicted by \LQD from $q$ after $t_q$
and no packet arriving to $q$ after $t_q$ is evicted by \LQD, since
an eviction in some step $t$ implies that the queue overflows at $t$.
We remark that we define $t_q = -1$ for queues $q$ that do not overflow in any step in order to have the property that for such queues, $t_q < t$ for all time steps $t$.

We would like to keep track of how many \OPT-extra packets are yet to be transmitted from a queue,
for which the following notation is useful. 
\begin{description}
\item{$e^t_q$:} the number of \OPT-extra packets transmitted from $q$ in step $t$ or later.%
\item{$\hat{e}^t_q = \max\{e^t_q-m_q,0\}$:} that is, $e^t_q$ adjusted for the packets that are canceled out by transmitted \LQD-extra packets.
	Note that $m_q$ will be specified in Section~\ref{sec:mapping}.\footnote{We remark that 
			defining $m_q$ is somewhat technical and not needed for the case $\LEXTRA = 0$. We have chosen to first present the core part of our analysis which is
	the scheme to split the \LQD profit.}
\end{description}

Note that $e_q = e^0_q = e^{t_q}_q$ as no \OPT-extra packet is transmitted before time $t_q$
by assumption~(A1). %
Thus, $e^t_q$ is constant up to time $t_q$.
After that, it further remains constant 
until $q$ becomes empty for \LQD, and then it decreases by one in each step until it becomes equal to zero.
The same property holds for $\hat{e}^t_q$.
The following useful observation follows from the fact that, by the definition of $t_q$, no packet destined to $q$ gets evicted from $q$ after $t_q$.

\begin{observation}\label{obs:t_q}
For any step $t$ and queue $q$ with $t\ge t_q$ (i.e., that does not overflow after~$t$), it holds that
$\max\left\{s^t_{\OPT}(q)-s^t_{\LQD}(q), 0\right\}\ge e^t_q$.
\end{observation}

See Fig.~\ref{fig:lifecycle} for an illustration of a life-cycle of a queue.

\begin{figure}[t]
	\begin{centering}
		\begin{tikzpicture}[scale=0.68]
			\newcommand\y{0}
			\draw[thick, ->] (-9,\y-4) -- (10,\y-4);
			\draw[thick, ->] (-9,\y-4) -- (-9,\y+0.5);
			\node[scale=0.6] () at (-9.8,\y+0.2) {packets};
			\node[scale=0.8] () at (-9.8,\y-3.75) {$1$};
			\node[scale=0.8] () at (-9.8,\y-3.25) {$2$};
			\node[scale=0.8] () at (-9.8,\y-2.75) {$3$};
			\node[scale=0.8] () at (-9.8,\y-2.25) {$4$};
			\node[scale=0.8] () at (-9.8,\y-1.75) {$5$};
			\node[scale=0.8] () at (-9.8,\y-1.25) {$6$};
			\node[scale=0.8] () at (-9.8,\y-0.75) {$7$};
			\node[scale=0.8] () at (-9.8,\y-0.25) {$8$};
			\node[scale=0.8] () at (9,\y-4.5) {time $t$};
			\draw[color=blue, pattern = north west lines, pattern color = blue]
			(-8,\y-4) -- (-8,\y-1.5) -- (-7.55,\y-1.5) -- (-7.55,\y-4);
			\draw[color = red, pattern = north east lines, pattern color = red]
			(-7.45,\y-4) -- (-7.45,\y-3.5) -- (-7,\y-3.5) -- (-7,\y-4);
			
			\newcommand\x{-6.5}
			\draw[thick,->] (\x-0.5,\y-3)--(\x,\y-3);
			\draw[color=blue, pattern = north west lines, pattern color = blue]
			(\x,\y-4) -- (\x,\y-2) -- (\x+0.45,\y-2) -- (\x+0.45,\y-4);
			\draw[color = red, pattern = north east lines, pattern color = red]
			(\x+0.55,\y-4) -- (\x+0.55,\y-3.5) -- (\x+1,\y-3.5) -- (\x+1,\y-4);
			\renewcommand\x{-5}
			\draw[thick,->] (\x-0.5,\y-3)--(\x,\y-3);
			\draw[color=blue, pattern = north west lines, pattern color = blue]
			(\x,\y-4) -- (\x,\y-1) -- (\x+0.45,\y-1) -- (\x+0.45,\y-4);
			\draw[color = red, pattern = north east lines, pattern color = red]
			(\x+0.55,\y-4) -- (\x+0.55,\y-3.5) -- (\x+1,\y-3.5) -- (\x+1,\y-4);
			\renewcommand\x{-3.5}
			\draw[thick,->] (\x-0.5,\y-3)--(\x,\y-3);
			\node () at (\x+0.5,\y-3) {\dots};
			\renewcommand\x{-2}
			\draw[thick,->] (\x-0.5,\y-3)--(\x,\y-3);
			\draw[color=blue, pattern = north west lines, pattern color = blue]
			(\x,\y-4) -- (\x,\y-1.5) -- (\x+0.45,\y-1.5) -- (\x+0.45,\y-4);
			\draw[color = red, pattern = north east lines, pattern color = red]
			(\x+0.55,\y-4) -- (\x+0.55,\y-3.5) -- (\x+1,\y-3.5) -- (\x+1,\y-4);
			
			\renewcommand\x{-0.5}
			\draw[thick,->] (\x-0.5,\y-3)--(\x,\y-3);
			\draw[color=blue, pattern = north west lines, pattern color = blue]
			(\x,\y-4) rectangle (\x+0.45,\y-2);
			\draw[color = red, pattern = north east lines, pattern color = red]
			(\x+0.55,\y-4) rectangle (\x+1,\y+0.5);
			\draw[decoration={calligraphic brace, raise=2pt}, thick, decorate] (\x+0.5,\y-1.9) -- (\x+0.5, \y+0.4);
			\node[] () at (\x-0.3,\y-0.85) {{\small $e_{q} \le$}};  %
			\node[] () at (\x+0.5,\y-4.5) {$t_{q}$};

			\renewcommand\x{1}
			\draw[thick,->] (\x-0.5,\y-3)--(\x,\y-3);
			\draw[color=blue, pattern = north west lines, pattern color = blue]
			(\x,\y-4) rectangle (\x+0.45,\y-2.5);
			\draw[color = red, pattern = north east lines, pattern color = red]
			(\x+0.55,\y-4) rectangle (\x+1,\y+0);
			\renewcommand\x{2.5}
			\draw[thick,->] (\x-0.5,\y-3)--(\x,\y-3);
			\draw[color=blue, pattern = north west lines, pattern color = blue]
			(\x,\y-4) rectangle (\x+0.45,\y-3);
			\draw[color = red, pattern = north east lines, pattern color = red]
			(\x+0.55,\y-4) rectangle (\x+1,\y-0.5);
			\renewcommand\x{4}
			\draw[thick,->] (\x-0.5,\y-3)--(\x,\y-3);
			\draw[color=blue, pattern = north west lines, pattern color = blue]
			(\x,\y-4) rectangle (\x+0.45,\y-3.5);
			\draw[color = red, pattern = north east lines, pattern color = red]
			(\x+0.55,\y-4) rectangle (\x+1,\y-1);
			\renewcommand\x{5.5}
			\draw[thick,->] (\x-0.5,\y-3)--(\x,\y-3);
			\draw[color = red, pattern = north east lines, pattern color = red]
			(\x+.55,\y-4) rectangle (\x+1,\y-1.5);
			\renewcommand\x{7}
			\draw[thick,->] (\x-0.5,\y-3)--(\x,\y-3);
			\node () at (\x+0.5,\y-3) {\dots};
			\renewcommand\x{8.5}
			\draw[thick,->] (\x-0.5,\y-3)--(\x,\y-3);
			\draw[color = red, pattern = north east lines, pattern color = red]
			(\x+.55,\y-4) rectangle (\x+1,\y-3.5);
			
		\end{tikzpicture}
		\caption{An example of a life-cycle of a queue $q$ that overflows in some steps,
			with $t_{q}$ being the last such step;
			the queue is depicted similarly as in Fig.~\ref{fig:buffer}.
			Note that $e_{q} = 5$.
			We remark that in ``hard instances'', \OPT would keep just one packet in $q$ before time $t_{q}$,
			while the size of the queue in the \LQD buffer varies. 
			Thus, the queue takes almost no space in the \OPT buffer,
			while it is larger in the \LQD buffer and both \OPT and \LQD gain packets from this queue in every step before $t_{q}$.
			At $t_{q}$, however, the situation reverses: \OPT stores many more packets than \LQD in the queue, 
			and if no packets arrive to this queue after $t_{q}$, \OPT gains a number of \OPT-extra packets.
			Note that $e_{q}$ is upper-bounded by the number of additional packets \OPT stores in $q$ compared to \LQD.
		}
		\label{fig:lifecycle}
	\end{centering}
\end{figure}

\subsection{Warm-Up: Proof of 2-Competitiveness} %
\label{sec:2comp}

\pvchanged{
Using the notation introduced above, we now show a simple proof of 2-competitiveness of \LQD.
Although this proof is still in essence the same as in the previous work~\cite{HahneKM01,AielloKM08}, 
our description is different, being less discrete in that it does not match \OPT-extra packets to specific packets transmitted by \LQD.
This new viewpoint allows us to identify slacks in the analysis and eventually exploit them to obtain a better bound.

Our main idea is that the \LQD profit in each step $t$ will be split among queues $q$ that do not overflow after $t$ proportionally to $e^t_q$, the number of \OPT-extra packets transmitted from $q$ at $t$ or later.
That is, letting $\LQD^t$ be the number of packets transmitted by \LQD at $t$,
we assign to any queue $q$ with $t\ge t_q$ a profit of $\LQD^t\cdot e^t_q / e^t$, where $e^t = \sum_{q:t\ge t_q} e^t_q$.
Clearly, the profit assigned at $t$ equals $\LQD^t$. Thus, it remains to show that each queue $q$ receives
a profit of at least $e_q$ in total over all steps,
which implies~\eqref{eqn:mainIneq} for $\rho = 1$, proving 2-competitiveness
as explained in Section~\ref{s:setup of analysis}.

To this end, we first derive an upper bound on $e^t$; in fact, a simple and possibly loose bound suffices.
Namely, we show that $e^t\le M$. Indeed, any queue $q$ with $t\ge t_q$ satisfies
$\max\left\{s^t_{\OPT}(q)-s^t_{\LQD}(q), 0\right\}\ge e^t_q$ by Observation~\ref{obs:t_q} and in particular, $s^t_{\OPT}(q)\ge e^t_q$. Summing up over all queues $q$ with $t\ge t_q$, we get 
$$M\ge \sum_{q:t\ge t_q} s^t_{\OPT}(q) \ge \sum_{q:t\ge t_q} e^t_q = e^t\,.$$
It follows that a queue $q$ with $t\ge t_q$ receives profit of at least $\LQD^t\cdot e^t_q / M$ at $t$.

Consider a queue $q$ with $e_q > 0$. At time $t_q$, queue $q$ overflows and thus, the \LQD buffer is full.
Furthermore, $s^{t_q}_{\LQD}(q) \ge s^{t_q}_{\max} - 1$. Let $s := s^{t_q}_{\LQD}(q)$. 
Note that the first \OPT-extra packet is transmitted from $q$ at $t_q+s$ or later and thus, for all times $t\in [t_q, t_q+s]$, 
$e^t_q = e_q$, implying that $q$ gets profit of at least $e_q / M$ from each packet transmitted by \LQD in these steps.
Moreover, observe that \LQD sends at least $M$ packets at times $t\in [t_q, t_q+s]$, since
the \LQD buffer is full at $t_q$ and $s^{t_q}_{\max} \le s+1$ (here, we also use that
packets are evicted from longest queues by \LQD). 
Hence, the total profit $q$ gets from steps $t\in [t_q, t_q+s]$ is at least $M\cdot e_q / M = e_q$, 
as desired.

\paragraph*{Intuition for improving upon ratio $2$.}
This simple analysis in fact has three sources of slack that we use to prove better than 2-competitiveness:
\begin{enumerate}[label=(S\arabic*)]
	\item \label{itm:slack_et_M}
		The inequality $e^t \le M$. Specifically, examining the derivation of $e^t \le M$ above,
	we observe that for any queue $q$ with $t\ge t_q$ and $e^t_q > 0$, Observation~\ref{obs:t_q} in fact implies $s^t_{\OPT}(q)\ge e^t_q + s^t_{\LQD}(q)$, but we only used $s^t_{\OPT}(q)\ge e^t_q$.
	\item For a fixed queue $q$ with $e_q > 0$, after step $t_q + s$ for $s = s^{t_q}_{\LQD}(q)$,
	there still may be \OPT-extra packets pending in the \OPT buffer, i.e., we may have $e^t_q > 0$ for $t > t_q + s$,
	so $q$ obtains some \LQD profit.
	\item We do not use \LQD-extra packets to cancel out some \OPT-extra packets.
\end{enumerate}
We note that these slacks in total may be negligible in that the preceding analysis does not give a ratio better than $2$
even when we consider them;
specifically this happens when $e^t$ is very close to $M$, $e_q$ is much smaller than $s = s^{t_q}_{\LQD}(q)$,
and there are no \LQD-extra packets.
However, combining the slacks and carefully modifying this simple scheme to split the \LQD profit,
we will eventually show that $q$ receives a profit of at least $\rho\cdot e_q$ for some $\rho > 1$.

Concretely, we modify this simple scheme of distributing the \LQD profit such that more profit is assigned to ``short'' queues $q$, i.e., those with $e_q \ll s = s^{t_q}_{\LQD}(q)$.
We choose a parameter $\alpha\in (0,1)$ 
and directly assign to $q$ a $(1-\alpha)$-fraction of the profit \LQD gains by transmitting packets from $q$ itself starting at time step $t_q$,
whereas the remaining $\alpha$-fraction of these packets is split proportionally to $\hat{e}^i_q$.
The parameter $\alpha\approx 0.62$ is chosen at the very end of the analysis, so as to minimize the competitive ratio upper bound.
See Fig.~\ref{fig:splitting} for an illustration.
However, analyzing this modified scheme is substantially more involved than the analysis for 2-competitiveness above,
requiring suitable lower bounds for the profit assigned proportionally.

In the following, we more formally describe this new modified scheme with all neccesary technical details.
}

\subsection{The Final Scheme to Split the \LQD Profit}

Before describing the scheme, we define time phases and introduce a potential
that allows for obtaining a better bound.

\paragraph*{Phases}
It will be convenient in certain parts of the analysis to consider time phases instead of time steps. More specifically, let $\tau_{1}<\tau_{2}<...<\tau_{\ell}$ be the time steps in which at least one queue overflows for the last time, i.e., for each $1\leq i\leq \ell$
there is a queue $q$ such that $\tau_{i}=t_{q}\ge 0$. Note that it has to be $\ell>0$ if \OPT gains extra profit (equivalently, $\ell=0$ only if $\EXTRA=0$). We call the time interval $[\tau_{i},\tau_{i+1})$ the \emph{$i$-th phase}; for $i=\ell$, we define $\tau_{\ell+1}=T+1$, where $T$ is the last time step during which any queue is active.
We remark that time steps before $\tau_{1}$ do not belong to any phase \pvchanged{(there are no \OPT-extra packets transmitted before step $\tau_{1}$,
so the \LQD's performance up to step $\tau_1$ is not worse than that of \OPT; moreover, 
$\tau_1$ may be the very first step of the instance).} Finally, observe that for any queue $q$ that overflows at least once (i.e., $t_q \ge 0$), there has to exist an $i$ such that $t_{q}=\tau_{i}$ as this queue overflows at $t_{q}$ for the last time.

In the remainder of the paper, our focus will be mainly on steps $\tau_1, \dots, \tau_\ell$.
For simplicity and to avoid double indexing, we shall write $s^i_{\LQD}(q)$ instead of $s^{\tau_i}_{\LQD}(q)$,
and similarly, we use index $i$ instead of $\tau_i$ in other notations.
Throughout the paper, $i$ will be used solely to index phases and time steps $\tau_1, \dots, \tau_\ell$.

\begin{figure}
	\begin{tikzpicture}[scale=0.83]
		\newcommand\y{0}
		\draw[very thick, ->] (-7,\y-4) -- (8.3,\y-4);
		\draw[very thick, ->] (-7,\y-4) -- (-7,\y+0.5);
		\node[scale=0.8] () at (-7.9,\y+0.2) {packets};
		\node[scale=0.8] () at (-7.8,\y-3.75) {$1$};
		\node[scale=0.8] () at (-7.8,\y-3.25) {$2$};
		\node[scale=0.8] () at (-7.8,\y-2.75) {$3$};
		\node[scale=0.8] () at (-7.8,\y-2.25) {$4$};
		\node[scale=0.8] () at (-7.8,\y-1.75) {$5$};
		\node[scale=0.8] () at (-7.8,\y-1.25) {$6$};
		\node[scale=0.8] () at (-7.8,\y-0.75) {$7$};
		\node[scale=0.8] () at (-7.8,\y-0.25) {$8$};
		\node[scale=0.8] () at (8,\y-4.5) {queues};
		\draw[very thick, dashed] (-6,\y-4) -- (-6,\y+0.5);
		\draw[very thick, dashed] (-5,\y-4) -- (-5,\y+0.5);
		\draw[very thick, dashed] (-4,\y-4) -- (-4,\y+0.5);
		\draw[very thick, dashed] (-3,\y-4) -- (-3,\y+0.5);
		\draw[very thick, dashed] (-2,\y-4) -- (-2,\y+0.5);
		\draw[very thick, dashed] (-1,\y-4) -- (-1,\y+0.5);
		\draw[very thick, dashed] (0,\y-4) -- (0,\y+0.5);
		\draw[very thick, dashed] (1,\y-4) -- (1,\y+0.5);
		\draw[very thick, dashed] (2,\y-4) -- (2,\y+0.5);
		\draw[very thick, dashed] (3,\y-4) -- (3,\y+0.5);
		\draw[very thick, dashed] (4,\y-4) -- (4,\y+0.5);
		\draw[very thick, dashed] (5,\y-4) -- (5,\y+0.5);
		\draw[very thick, dashed] (6,\y-4) -- (6,\y+0.5);
		\draw[very thick, dashed] (7,\y-4) -- (7,\y+0.5);
		\draw[color=blue, pattern = north west lines, pattern color = blue]
		(-7,\y-4) -- (-7,\y-1.5) -- (-6.55,\y-1.5) -- (-6.55,\y-4) --
		(-6,\y-4) -- (-6,\y-1.5) -- (-5.55,\y-1.5) -- (-5.55,\y-4) --
		(-5,\y-4) -- (-5,\y-1.5) -- (-4.55,\y-1.5) -- (-4.55,\y-4) --
		(-4,\y-4) -- (-4,\y-1.5) -- (-3.55,\y-1.5) -- (-3.55,\y-4) --
		(-3,\y-4) -- (-3,\y-1.5) -- (-2.55,\y-1.5) -- (-2.55,\y-4) --
		(-2,\y-4) -- (-2,\y-2.0) -- (-1.55,\y-2.0) -- (-1.55,\y-4) --
		(-1,\y-4) -- (-1,\y-2.0) -- (-0.55,\y-2.0) -- (-0.55,\y-4) --
		( 0,\y-4) -- ( 0,\y-2.0) -- ( 0.45,\y-2.0) -- ( 0.45,\y-4) --
		( 1,\y-4) -- ( 1,\y-2.5) -- ( 1.45,\y-2.5) -- ( 1.45,\y-4) --
		( 2,\y-4) -- ( 2,\y-3.0) -- ( 2.45,\y-3.0) -- ( 2.45,\y-4) --
		( 3,\y-4) -- ( 3,\y-3.5) -- ( 3.45,\y-3.5) -- ( 3.45,\y-4);
		\draw[color = red, pattern = north east lines, pattern color = red]
		(-6.45,\y-4) -- (-6.45,\y-3.5) -- (-6,\y-3.5) -- (-6,\y-4) --
		(-5.45,\y-4) -- (-5.45,\y-3.5) -- (-5,\y-3.5) -- (-5,\y-4) --
		(-4.45,\y-4) -- (-4.45,\y-3.5) -- (-4,\y-3.5) -- (-4,\y-4) --
		(-3.45,\y-4) -- (-3.45,\y-3.5) -- (-3,\y-3.5) -- (-3,\y-4) --
		(-2.45,\y-4) -- (-2.45,\y-3.5) -- (-2,\y-3.5) -- (-2,\y-4) --
		(-1.45,\y-4) -- (-1.45,\y-3.5) -- (-1,\y-3.5) -- (-1,\y-4) --
		(-0.45,\y-4) -- (-0.45,\y-3.5) -- (0,\y-3.5) -- (0,\y-4) --
		(0.55,\y-4) -- (0.55,\y+0.0) -- (1,\y+0.0) -- (1,\y-4) --
		(1.55,\y-4) -- (1.55,\y-0.5) -- (2,\y-0.5) -- (2,\y-4) --
		(2.55,\y-4) -- (2.55,\y-1.0) -- (3,\y-1.0) -- (3,\y-4) --
		(3.55,\y-4) -- (3.55,\y-1.5) -- (4,\y-1.5) -- (4,\y-4) --
		(4.55,\y-4) -- (4.55,\y-2.0) -- (5,\y-2.0) -- (5,\y-4) --
		(5.55,\y-4) -- (5.55,\y-2.5) -- (6,\y-2.5) -- (6,\y-4) --
		(6.55,\y-4) -- (6.55,\y-3.0) -- (7,\y-3.0) -- (7,\y-4) --
		(7.55,\y-4) -- (7.55,\y-3.5) -- (8,\y-3.5) -- (8,\y-4);
		\draw[decoration={calligraphic brace,raise=2pt}, thick, decorate] (-0.2,\y-4) --  (-7,\y-4);
		\node[] () at (-3.6, \y-4.5) {{\small LQD{} profit split equally}}; %
		\draw[decoration={calligraphic brace, raise=2pt}, thick, decorate] (4,\y-4) -- (0, \y-4);
		\node[] () at (2,\y-4.5) {{\small $\alpha$-fraction to the same queue}};  %
		\node[] () at (2,\y-5) {{\small $(1-\alpha)$-fraction split equally}};
\end{tikzpicture}
\caption{Intuition for splitting the \LQD profit in some step $\tau_i$;
	the queues are depicted similarly as in Fig.~\ref{fig:buffer}.
	The queues $q$ on the left, in which \OPT stores just one packet,
	typically satisfy $\tau_i < t_q$, i.e., they overflow after $\tau_i$.
	On the other hand, queues on the right, where \OPT stores more packets than \LQD,
	do not overflow after $\tau_i$ and \OPT transmits some \OPT-extra packets from them.
We note that in hard instances, the maximum size of a queue for \LQD, denoted $s^t_{\max}$,
may fluctuate over time, as witnessed in the lower bound of $\approx 1.44$ in~\cite{bochkov2019new}.}
\label{fig:splitting}
\end{figure}

\paragraph*{Potential}
We introduce a potential that will help
us to deal with the fact that some queues overflowing in step $t$ may only have $s^t_{\max} - 1$ packets and not $s^t_{\max}$ packets.
On an intuitive level, this potential amortizes the profit assignment by moving some profit from phases with a slack
to phases in which our lower bounds on the profit assigned are tight; we develop these bounds in the subsequent sections.

Namely, at any phase $i$, let $\calA^i$ be the set of queues $q$ that are active in step $\tau_i$ and satisfy $t_q > \tau_i$
(i.e., will overflow after the beginning of phase $i$).
Thus, for any such queue $q$ we have $t_q = \tau_j$ for some $j > i$
and consequently, $q$ is non-empty for \LQD in every step during phase $i$ by assumption~(A1) %
(as otherwise, $q$ would not overflow at $\tau_j$).

Then, using the aforementioned parameter $\alpha\in (0,1)$, we define potential %
$\displaystyle \Psi^i := \alpha\cdot |\calA^i|$.
Note that the potential at the beginning is $\Psi^1 \le \alpha\cdot M$ and after the last packet of the input instance is transmitted, the potential equals $\Psi^{\ell+1} = 0$.
We define two quantities which express the change of this potential in phase $i$:
\begin{description}
\item{$u^i = $}  the number of queues active in step $\tau_{i+1}$ that were inactive in step $\tau_i$
	and will overflow after $\tau_{i+1}$, i.e., the number of ``new'' active queues that will overflow after $\tau_{i+1}$; and
\item{$v^i = $} the number of queues that are active in step $\tau_i$ and overflow at $\tau_{i+1}$ for the last time,
i.e., $\tau_{i+1} = t_q$ for any such queue $q$.
\end{description}
Let $\Delta^i \Psi := \Psi^{i+1} - \Psi^i$ be the change of the potential in phase $i$;
observe that $\Delta^i \Psi = \alpha\cdot (u^i - v^i)$.

\paragraph*{Splitting the \LQD profit}
We now formally define our scheme of splitting the \LQD profit such that we assign a profit of at least 
$\rho\cdot \hat{e}_q$ to each $q$. 
To keep track of how much profit we assigned to a queue $q$, we use a counter $\Phi_q$.
In particular, $\Delta^i \Phi_q$ will be the \LQD profit assigned to $q$ in phase $i$ and
$\Phi_q = \sum_{i=1}^\ell \Delta^i \Phi_q$ will be the \LQD profit assigned to $q$ over all phases.
Let $\LQD^i$ be the profit
of \LQD in phase $i$, i.e., the total number of packets transmitted by \LQD in all time steps in $[\tau_i, \tau_{i+1})$.
\pvchanged{
We will ensure that the profit assigned to all queues in phase $i$ is at most $\LQD^i - \Delta^i \Psi$;
summing over all phases $i$, the total profit that we assign is at most $\LQD + \Psi^1 - \Psi^{\ell+1} \le \LQD +\alpha\cdot M$.
}

The crucial part will be to show that, for every queue $q$, $\Phi_q\ge \rho\cdot \hat{e}_q$,
which, together with $\sum_{q} (e_q - \hat{e}_q)= \sum_q m_q \le \LEXTRA$, implies~\eqref{eqn:mainIneq} using
$$
\LQD+\alpha\cdot M \ge \sum_i \LQD^i - \Delta^i \Psi \ge \sum_{q} \Phi_q \ge \sum_{q} \rho\cdot \hat{e}_q
\ge \rho\cdot\left(\sum_{q} e_q - \LEXTRA\right)\,.
$$

Consider phase $i$.
Our first idea is to split $\LQD^{i}- \Delta^i \Psi$
among queues $q$ satisfying $t_q \le \tau_i$ proportionally to $\hat{e}^i_q$,
meaning that we assign a profit of $(\LQD^i- \Delta^i \Psi)\cdot \hat{e}^i_q / \hat{e}^i$ to a queue $q$ with $\tau_i\ge t_q$,
where $\hat{e}^i = \sum_{q: \tau_i\ge t_q} \hat{e}^i_q$. 
Such a scheme is useful because we can relate $\hat{e}^i$ to a certain fraction of the \LQD profit; this is elaborated in Section~\ref{sec:mapping}.
However, as shown in Section~\ref{sec:2comp}, it only proves 2-competitiveness \pvchanged{(the possibly positive contribution of $- \Delta^i \Psi$ may be negligible)}.

We therefore split $\LQD^i$ into two parts: 
 Let $o^i$ be the number of packets that \LQD transmits in the $i$-th phase from queues $q$ with $\tau_i\ge t_q$ and $e_q > 0$,
and let $n^i$ be the number of packets transmitted by \LQD in phase $i$
from all  other queues, i.e., from queues that either overflow after phase $i$ or no \OPT-extra packet
is transmitted from them.
Note that $\LQD^i = o^i + n^i$.

Given a parameter $\alpha\in (0,1)$, in each phase ${i}$ with $\hat{e}^i > 0$, we assign an \LQD profit of
\begin{equation}\label{eqn:splitting profit def init}
\Delta^{i}\Phi_q := \frac{\hat{e}^i_q}{\hat{e}^i} \cdot (n^{i} + \alpha \cdot o^{i} - \Delta^i \Psi)
 \,\, + \,\, (1 - \alpha)\cdot o^i_q
\end{equation}
to each queue $q$ with $\tau_i\ge t_q$ and $\hat{e}^i_q>0$, where $o^i_q$ is the number of packets that \LQD transmits from $q$
during the $i$-th phase \pvchanged{(we note that $\hat{e}^i_q>0$ implies that $e_q > 0$, so $o^i_q$ packets sent from $q$ by \LQD 
in the $i$-th phase are accounted for in $o^i$)}. 
Note that we only assign profit to queues that already
have overflown for the last time. Furthermore, once a queue $q$ with $\tau_i\ge t_q$ is empty in both the \LQD and \OPT buffers at the start of a phase,
it does not get any profit as $\hat{e}^i_q\le e^i_q = 0$ and $o^i_q = 0$.

To ensure feasibility of our scheme, we show that in total over all queues $q$ with $\tau_i\ge t_q$ and $\hat{e}^i_q>0$
we assign a profit of at most $\LQD^i - \Delta^i \Psi$ in phase $i$.
Indeed, using $\hat{e}^i = \sum_{q: \tau_i\ge t_q} \hat{e}^i_q$ and $\sum_{q: \tau_i\ge t_q \text{\ and\ } \hat{e}^i_q>0} o^i_q \le o^i$, we have
\begin{align*}
	\sum_{q: \tau_i\ge t_q \text{\ and\ } \hat{e}^i_q>0} \mkern-29mu\Delta^i \Phi_q
	&= \sum_{q: \tau_i\ge t_q \text{\ and\ } \hat{e}^i_q>0} \left(\,\frac{\hat{e}^i_q}{\hat{e}^i} \cdot (n^{i} + \alpha \cdot o^{i} - \Delta^i \Psi)
	+ (1 - \alpha)\cdot o^i_q\,\right)
	\\
	&\le n^{i} + \alpha \cdot o^{i} - \Delta^i \Psi + (1 - \alpha)\cdot o^i
	\\
	&= n^i+o^i - \Delta^i \Psi = \LQD^i - \Delta^i \Psi\,.
\end{align*}

While the scheme to split the \LQD profit is relatively simple to define,
showing $\Phi_q\ge \rho\cdot \hat{e}_q$ brings technical challenges, namely,
in obtaining suitable lower bounds on the profits assigned proportionally to $\hat{e}^i_q$
and in summing up these lower bounds over all phases.
We get our lower bound based on a novel scheme that relates the buffers of \LQD and of \OPT,
which is introduced in the next two sections.

\section{Live and Let Die}
\label{s:liveAndDying}

We start by deriving a helpful lower bound on 
$n^i + \alpha \cdot o^i - \Delta^i \Psi$. For this, we first introduce the notion of live and dying queues,
which are defined with respect to a fixed queue $q$ with $t_q\le \tau_i$ and $\hat{e}_q > 0$.
For this fixed queue, we need to define live and dying queues up until the first phase that comes after \OPT transmits the last \OPT-extra
packet from $q$ not canceled out by an \LQD-extra packet. Let $j_q := \min\{j : \hat{e}^j_q = 0\}$ be the index $j$ of the earliest step $\tau_j$ in which all remaining \OPT-extra packets to be transmitted from $q$ (if any) are canceled out; as $\hat{e}_q > 0$, index $j_q$ is well-defined (possibly $j_q = \ell+1$, which ends at the last step $\tau_{\ell+1}-1=T$ when a queue is active). %

\begin{definition}\label{def:liveDying}
Fix a queue $q$ with $\hat{e}_q > 0$,
and consider a phase $i$ with $t_q\le \tau_i$ and $i \le j_q$.
Let $q'$ be a queue for which \LQD stores at least one packet at time step $\tau_i$. Queue $q'$ is called \emph{live} with respect to (w.r.t.) queue $q$ at time step $\tau_i$ if
\begin{enumerate}[label=(\roman*)]
\item $\tau_i<t_{q'}$, i.e., $q'$ overflows at some time step after the $i$-th phase, or
\item $e_{q'}=0$ and $s^{j_q}_{\LQD}(q') \ge 1$, 
\end{enumerate}
or both. Otherwise, $q'$ is called \emph{dying} with respect to queue $q$ at time $\tau_i$.
\end{definition}
Note that step $\tau_{j_q}$ referred to in $s^{j_q}_{\LQD}(q')$ is after $t_q$, since $\hat{e}^t_q = \hat{e}_q > 0$
in any step $t$ before the first \OPT-extra packet is transmitted from $q$.
Furthermore, $e_{q'}>0$ implies that $q'$ becomes empty in the \LQD buffer before it becomes empty in the \OPT buffer,
by Observation~\ref{obs:t_q}.
Intuitively, and assuming that $\hat{e}_q = e_q$, at time step
$\tau_i$, a queue $q'$ is dying with respect to $q$ if (i) it no longer
overflows and (ii) either \LQD runs out of packets
to send from $q'$ before the time step \OPT does 
or \LQD runs out of packets to send from $q'$ before the beginning of phase $j_q$.

The definition of live and dying queues implies the following property about transitions between these two types.
This follows since the only property in Definition~\ref{def:liveDying} (for a fixed $q$) that may change with increasing $i$
is whether or not $\tau_i < t_{q'}$.

\begin{observation}\label{obs:liveDying}
If a queue $q'$ is dying (w.r.t.\ queue $q$) in time step $\tau_i$, it will never be live (w.r.t.\ queue $q$) in step $\tau_j$ for any $j > i$.
If $q'$ is live (w.r.t.\ queue $q$) at $\tau_i$, it can become dying (w.r.t.\ queue $q$) in time step $\tau_{i+1}$ only if it overflows in time step $\tau_{i+1}$ for the last time.
\end{observation}

For any phase $i$, we denote the set of queues that are live in step $\tau_i$ w.r.t.\ $q$ as $\calL^i_q$
and the set of queues dying in step $\tau_i$ w.r.t.\ $q$ as $\calD^i_q$.
For a fixed phase $i$ and queue $q$, the sets $\calL^i_q$ and $\calD^i_q$ partition all queues
in which \LQD stores packets at time $\tau_i$.
Let $d^i_q$ be the number of packets transmitted by \LQD from queues in $\calD^i_q$ during phase $i$.
To get a lower bound on $\Delta^{i}\Phi_q$ defined in~\eqref{eqn:splitting profit def init}, we now relate $n^i + \alpha\cdot o^i$ to $|\calL^i_q|$ and $d^i_q$.

\begin{observation}\label{obs:live_and_dying}
	For any queue $q$,
it holds that $n^i\ge |\calL^i_q|\cdot (\tau_{i+1} - \tau_i)$ %
and also $n^i + \alpha\cdot o^i \ge |\calL^i_q|\cdot (\tau_{i+1} - \tau_i) + \alpha\cdot d^i_q$.
\end{observation}
\begin{proof}
Recall that $o^i$ is the number of packets that \LQD transmits in phase $i$ from queues $q'$
satisfying $\tau_i\ge t_{q'}$ and $e_{q'}>0$, and that $n^i = \LQD^i - o^i$
(i.e., $n^i$ is the number of packets that \LQD sends in the $i$-th phase from queues $q'$ that will
overflow after $\tau_i$ or that satisfy $e_{q'} = 0$). %
As packets sent from queues that are live w.r.t.\ queue $q$ in step $\tau_i$ are accounted for in $n^i$,
it holds that $n^i\ge |\calL^i_q|\cdot (\tau_{i+1} - \tau_i)$,
which proves the first claim.

Since every queue $q'$ with $\tau_i\ge t_{q'}$ and $e_{q'}>0$ is dying w.r.t.\ queue $q$ at time step $\tau_i$, we have that $o^i \le d^i_q$.
It holds that $|\calL^i_q|\cdot (\tau_{i+1} - \tau_i) + d^i_q \le \LQD^i = n^i + o^i$, %
and this inequality
implies the second claim by using $o^i \le d^i_q$ and $\alpha \le 1$.
\end{proof}

Fix a queue $q$. We would like to lower bound the number $d^i_q$ of packets transmitted from dying queues during the $i$-th phase in some way. Note that the \LQD buffer is full at times $\tau_i$ and $\tau_{i+1}$. Suppose for a moment that the set of live queues (w.r.t.\ queue $q$) does not change between step $\tau_i$ and step $\tau_{i+1}$, i.e., $\calL^{i+1}_q=\calL^i_q$. Now, if the number of packets that \LQD stores in live queues $\calL^i_q$ increases by $m$ between step $\tau_i$ and step $\tau_{i+1}$, then we know that $d^i_q\ge m$. This is because the buffer is full, so if the live queues gain $m$ packets, then dying queues must have lost at least $m$ packets (possibly more if there are new dying queues in step $\tau_{i+1}$). Since dying queues do not overflow, the only possible way to reduce the number of packets stored by \LQD in dying queues is to transmit them.

We now formalize this intuition and handle cases where the set of live queues changes from one phase to the next.
For the fixed queue $q$ and each phase $i$ such that $\tau_i\ge t_q$ and $i\le j_q$, we define
\begin{equation}\label{eqn:sigma_def}
\sigma^i_q = 
\begin{cases}
\left(\sum_{q'\in \calL^i_q} s^i_{\LQD}(q')\right) / |\calL^i_q|, &\text{if\ }|\calL^i_q|\ge 1,
\\
1,& \text{otherwise}\,. 
\end{cases}
\end{equation}
In words, $\sigma^i_q$ equals the average number of \LQD packets in live queues (w.r.t.\ queue $q$) in step $\tau_i$,
provided that there is at least one such queue.
Since live queues are non-empty for \LQD, it holds that $\sigma^i_q\ge 1$.
Furthermore, as the average is at most the maximum and as the maximum is an integer, this gives us the following observation.
\begin{observation}\label{obs:sigma_upper_bound}
In any phase $i$ such that $\tau_i\ge t_q$ and $i\le j_q$, it holds that $\lceil\sigma^i_q\rceil\le s^i_{\max}$.
\end{observation}

We will also need that $\sigma^i_q\ge 2$ once there is at least one live queue (w.r.t.\ queue $q$) in phase $i$.

\begin{observation}\label{obs:sigma_lower_bound}
For each $i$ with $\tau_i\ge t_q$, $i < j_q$, and $|\calL^i_q|\ge 1$, it holds that $\sigma^i_q\ge 2$.
\end{observation}
\begin{proof}
We show that any live queue $q'$ (w.r.t.\ queue $q$) has at least two packets in the \LQD buffer in any step $\tau_i\ge t_q$ with $i < j_q$, which is sufficient as $\sigma^i_q$ is the average size of live queues, provided that $|\calL^i_q|\ge 1$.
For a live queue $q'$, consider two cases (as in Definition~\ref{def:liveDying}):

First, if $\tau_i < t_{q'}$ we show $s^i_{\LQD}(q') \ge 2$ by a contradiction: %
If we had $s^i_{\LQD}(q') \le 1$, then $q'$ would be empty at $t_{q'}$, no packet would arrive to $q'$ at $t_{q'}$ by assumption~(A1), and thus, $q'$ would not overflow in that step.
Second, if $e_{q'} = 0$ and $s^{j_q}_{\LQD}(q') \ge 1$, where $j_q = \min\{j : \hat{e}^j_q = 0\}$,
then we have that $s^i_{\LQD}(q')\ge 2$, using assumption~(A1) again together with $i < j_q$. %
\end{proof}

\paragraph*{Packets transmitted from dying queues}
We can now formally state our lower bound on the number of packets transmitted from dying queues, taking into account
the change of the potential as well.
As a byproduct (by rearranging the bound on $d^i_q$ below), we obtain an upper bound on $u^i$,
the number of ``new'' active queues that will overflow after $\tau_{i+1}$, which captures the increase of the potential.
Recall that $v^i$ equals the number of queues that are active in step $\tau_i$ and overflow at $\tau_{i+1}$ for the last time.
Note that the following bound is only useful when the value $\sigma_q$ increases during phase $i$, i.e., $\sigma^{i+1}_q > \sigma^i_q$,
or there are relatively many ``new'' active queues that overflow after $\tau_{i+1}$.

\begin{lemma}\label{lem:dying_bound}
Consider any queue $q$ with $\hat{e}_q > 0$.
For each phase $i$ with $\tau_i\ge t_q$ and $\hat{e}^i_q > 0$, %
it holds that
$d^i_q \ge (\sigma^{i+1}_q - \sigma^i_q)\cdot |\calL^i_q| + \sigma^{i+1}_q\cdot u^i - v^i$.
\end{lemma}

\begin{proof}
As $q$ is fixed, we consider live and dying queues w.r.t.\ queue $q$ only.
For simplicity, let $\tau = \tau_i$ and $\tau' = \tau_{i+1}$, thus the $i$-th phase is $[\tau, \tau')$.
By the definition of $\sigma^i_q$, live queues contain $\sigma^i_q\cdot |\calL^i_q|$ packets in the \LQD buffer in step $\tau$
and thus, dying queues $\calD^i_q$ have $M - \sigma^i_q\cdot |\calL^i_q|$ packets in total at $\tau$, since the \LQD buffer is full in step $\tau$.
As dying queues do not overflow in any step after $\tau$, it is sufficient
to show that queues $\calD^i_q$ altogether contain at most
$M - \sigma^i_q\cdot |\calL^i_q| - (\sigma^{i+1}_q - \sigma^i_q)\cdot |\calL^i_q| - \sigma^{i+1}_q\cdot u^i + v^i = M - \sigma^{i+1}_q\cdot |\calL^i_q| - \sigma^{i+1}_q\cdot u^i + v^i$
packets in \LQD's buffer in step $\tau'$. Let $x$ be the number of \LQD packets in queues $\calD^i_q$ in step $\tau'$,
so our goal is to show
\begin{equation}\label{eqn:dying bound goal}
x \le M - \sigma^{i+1}_q\cdot |\calL^i_q| - \sigma^{i+1}_q\cdot u^i + v^i\,.
\end{equation}

To this end, we analyze the \LQD buffer in step $\tau'$.
By Observation~\ref{obs:liveDying}, any live queue in $\calL^i_q$ is also live in step $\tau'$
or overflows in step $\tau'$ (possibly both).
Let $\calL' = \calL^i_q\cup \calL^{i+1}_q$ be the set of queues that are live in step $\tau$ or in step $\tau'$.
Observe that $\calL'\cap \calD^i_q = \emptyset$,
since by Observation~\ref{obs:liveDying} dying queues may only become empty in \LQD's buffer but not live.
It follows that queues in $\calL'\setminus \calL^i_q$ must be inactive in step $\tau$. 
Next, no queue that is live in step $\tau$ is empty for \LQD in step $\tau' = \tau_{i+1}$ by Definition~\ref{def:liveDying},
using $i < j_q$, which follows from $\hat{e}^i_q > 0$.
Finally, note that $\calL'$ may contain some dying queues in $\calD^{i+1}_q$, but all of them must overflow at $\tau'$,
and that $\calL'\setminus \calL^{i+1}_q \subseteq \calD^{i+1}_{q}\setminus \calD^i_{q}$.
Concluding, set $\calL'$ consists of three disjoint types of queues: %
\begin{enumerate}[nosep,label=(\roman*)]
\item live queues in step $\tau$ that remain live in step $\tau'$, i.e., $ \calL^i_q\cap \calL^{i+1}_q$,
\item live queues in step $\tau$ that become dying in step $\tau'$ --- these are queues in $\calL'\setminus \calL^{i+1}_q$
and we have that $|\calL'\setminus \calL^{i+1}_q| = v^i$, which follows from Observation~\ref{obs:liveDying}
and from the definition of $v^i$, and 
\item queues inactive in step $\tau$ that are live in step $\tau'$ --- these are queues in $\calL'\setminus \calL^i_q$
and there are at least $u^i$ many of them (some live queues may not overflow in any step, so they are not accounted for in $u^i$).
\end{enumerate}
See Figure~\ref{fig:live_die} for an illustration.

\begin{figure}[t]
\begin{centering}
\begin{tikzpicture}[scale=0.92,level/.style={},decoration={brace,mirror,amplitude=7}]
  \draw[color = red, pattern = north east lines, pattern color = red]
  (0,0) -- (0, 0.5) -- (1,0.5) -- (1,1.5) -- (2,1.5) -- (2,2.5) -- (3,2.5) -- (3,3) -- (4,3) -- (4,0);
  \filldraw[color=green!10!black, fill = green!10!white ]
  (4,0) -- (4,4.5) -- (5,4.5) -- (5,0);
  \draw[color=blue, pattern = north west lines, pattern color = blue ]
  (5,0) -- (5,5) -- (7,5) -- (7,4.5) -- (9,4.5) -- (9,3) -- (10,3) -- (10, 2.5) -- (11, 2.5) -- (11, 3.5) -- (12, 3.5) -- (12,0);
  \draw[very thick, ->] (0,0) -- (13,0);
  \draw[very thick, ->] (0,0) -- (0,5.5);
  \node () at (0.5,-0.5) {$\downarrow$ $1$};
  \node () at (1.5,-0.5) {$\downarrow$ $2$};
  \node () at (2.5,-0.5) {$\downarrow$ $3$};
  \node () at (3.5,-0.5) {$\downarrow$ $4$}; 
  \node () at (4.5,-0.5) {$\downarrow$ $5$};
  \node () at (5.5,-0.5) {$\downarrow$ $6$};
  \node () at (6.5,-0.5) {$\downarrow$ $7$};
  \node () at (7.5,-0.5) {$\downarrow$ $8$};
  \node () at (8.5,-0.5) {$\downarrow$ $9$};
  \node () at (9.5,-0.5) {$\downarrow$ $10$};
  \node () at (10.5,-0.5) {$\downarrow$ $11$};
  \node () at (11.5,-0.5) {$\downarrow$ $12$};
  \node () at (13,-0.5) {queues};
  \draw[very thick, dashed] (1,-0.5) -- (1,5.5);
  \draw[very thick, dashed] (2,-0.5) -- (2,5.5);
  \draw[very thick, dashed] (3,-0.5) -- (3,5.5);
  \draw[very thick, dashed] (4,-0.5) -- (4,5.5);
  \draw[very thick, dashed] (5,-0.5) -- (5,5.5);  
  \draw[very thick, dashed] (6,-0.5) -- (6,5.5);
  \draw[very thick, dashed] (7,-0.5) -- (7,5.5);
  \draw[very thick, dashed] (8,-0.5) -- (8,5.5);
  \draw[very thick, dashed] (9,-0.5) -- (9,5.5);
  \draw[very thick, dashed] (10,-0.5) -- (10,5.5);
  \draw[very thick, dashed] (11,-0.5) -- (11,5.5);
  \draw[very thick, dashed] (12,-0.5) -- (12,5.5);
  \node () at (-0.8,5.2) {packets};
  \node () at (-0.4,0.25) {$1$};
  \node () at (-0.4,0.75) {$2$};
  \node () at (-0.4,1.25) {$3$};
  \node () at (-0.4,1.75) {$4$};
  \node () at (-0.4,2.25) {$5$};
  \node () at (-0.4,2.75) {$6$};
  \node () at (-0.4,3.25) {$7$};
  \node () at (-0.4,3.75) {$8$};
  \node () at (-0.4,4.25) {$9$};
  \node () at (-0.4,4.75) {$10$};
  
  \draw [decorate] (0,-0.7) --node[below=2mm]{$\calD^{i+1}_q$} (5,-0.7);
  \draw [decorate] (5,-0.7) --node[below=2mm]{$\calL^{i+1}_q$} (12,-0.7);
  \draw [decorate] (0,-1.5) --node[below=2mm]{$\calD^i_q$} (4,-1.5);
  \draw [decorate] (4,-1.5) --node[below=2mm]{$\calL^i_q$} (11,-1.5);
\end{tikzpicture}
\caption{An example of the \LQD buffer in step $\tau' = \tau_{i+1}$ for illustrating the proof of Lemma~\ref{lem:dying_bound}.
	Note that $s^{i+1}_{\max} = 10$ and that
	queues $5-9$ overflow (and thus the \LQD buffer is full).
	However, only queue~5 becomes dying as it overflows for the last time at $\tau'$, i.e., $t_5 = \tau'$.	
	Moreover, queue~12 was empty in step $\tau_i$ and queue~1 will become empty for \LQD just after packets are transmitted in step $\tau'$ (note that no further packets will arrive to queue~1 after step $\tau'$ by assumption~(A1)). %
	We have that $\calL' = \{5, 6, \dots, 12\}$.
	Finally, $\sigma^{i+1}_q = 56/7 = 8$, since %
	there are 56 packets in 7 live queues $\calL^{i+1}_q$.
	}
\label{fig:live_die}
\end{centering}
\end{figure}

Any queue in $\calL' \setminus\calL^{i+1}_q$ must overflow at $\tau'$,
so it has at least $s^{i+1}_{\max} - 1\ge \sigma^{i+1}_q - 1$ \LQD packets at $\tau'$, where we use $\sigma^{i+1}_q\le s^{i+1}_{\max}$
by Observation~\ref{obs:sigma_upper_bound}. 
It follows that queues in $\calD^{i+1}_q$ have at least
$x + |\calL' \setminus\calL^{i+1}_q|\cdot (\sigma^{i+1}_q - 1)$ packets in total in step $\tau'$.
By the definition of $\sigma^{i+1}_q$ and since \LQD's buffer is full at $\tau'$,
queues in $\calD^{i+1}_q$ contain $M - |\calL^{i+1}_q|\cdot \sigma^{i+1}_q$ packets
and thus $$x + |\calL' \setminus\calL^{i+1}_q|\cdot (\sigma^{i+1}_q - 1) \le M - |\calL^{i+1}_q|\cdot \sigma^{i+1}_q\,.$$
Rearranging and using $|\calL'\setminus \calL^{i+1}_q| = v^i$, we get
$x \le M - |\calL'|\cdot \sigma^{i+1}_q + v^i$. 
Using $|\calL'| = |\calL^i_q| + |\calL'\setminus \calL^i_q| \ge |\calL^i_q| + u^i$,
we obtain $x \le M - (|\calL^i_q| + u^i)\cdot \sigma^{i+1}_q + v^i$, implying~\eqref{eqn:dying bound goal}.
This concludes the proof as explained above.
\end{proof}

We now give two more upper bounds on $u^i$, the number of ``new'' active queues that will overflow after $\tau_{i+1}$.
The advantage of the following bound over the one from Lemma~\ref{lem:dying_bound} is that 
it does not use $\sigma^{i+1}_q$, i.e., it is suitable for phase $i = j_q - 1$.

\begin{lemma}\label{lem:new active UB first}
Consider any queue $q$ with $\hat{e}_q > 0$.
For each phase $i$ with $\tau_i\ge t_q$ and $\hat{e}^i_q > 0$,
it holds that $u^i \le \frac12\cdot \left(|\calL^i_q|\cdot (\sigma^i_q - 1) + d^i_q\right)$.
\end{lemma}

\begin{proof}
As $q$ is fixed, we consider live and dying queues w.r.t.\ queue $q$ only.
Let $x$ be the number of \LQD packets in queues $\calD^i_q$ in step $\tau_{i+1}$.
Similarly as in the proof of Lemma~\ref{lem:dying_bound}, we show that
\begin{equation}\label{eqn:new active UB first goal}
x \le M - |\calL^i_q| - 2\cdot u^i\,.
\end{equation}
This equation implies the lemma, since dying queues $\calD^i_q$ have $M - \sigma^i_q\cdot |\calL^i_q|$ packets in total in step $\tau_i$
and thus, $d^i_q = M - \sigma^i_q\cdot |\calL^i_q| - x \ge M - \sigma^i_q\cdot |\calL^i_q| - (M - |\calL^i_q| - 2\cdot u^i)
= 2\cdot u^i - (\sigma^i_q - 1)\cdot |\calL^i_q|$
by~\eqref{eqn:new active UB first goal}, from which the lemma follows by rearranging.
To justify~\eqref{eqn:new active UB first goal}, 
queues accounted for in $u^i$ are empty for \LQD at $\tau_i$ and overflow after $\tau_{i+1}$, so \LQD must store at least two packets in each of them in step $\tau_{i+1}$ by assumption~(A1).
Moreover, no queue that is live in step $\tau_i$ is empty for \LQD in step $\tau_{i+1}$ by Definition~\ref{def:liveDying},
using $i < j_q$, which follows from $\hat{e}^i_q > 0$.
Hence, there can be at most $M - |\calL^i_q| - 2\cdot u^i$ \LQD packets in queues $\calD^i_q$ in step $\tau_{i+1}$.
\end{proof}

Finally, we give a third upper bound on $u^i$ that is incomparable to those in Lemmas~\ref{lem:dying_bound}
and~\ref{lem:new active UB first} and is useful for phases with a relatively small number of steps.

\begin{lemma}\label{lem:new active UB simple}
For any phase $i$, it holds that $u^i \le \frac12\cdot (n^i + o^i)$.
\end{lemma}

\begin{proof}
Recall that the new active queues accounted for in $u^i$ are empty in both buffers in step $\tau_i$ and will
overflow after $\tau_{i+1}$; let $\calU$ be the set of these $u^i$ queues.
We show that the number of packets in queues $\calU$ in step $\tau_{i+1}$ is at most $n^i + o^i$, i.e.,
the number of packets \LQD transmits during the $i$-th phase.
This is sufficient, since any queue in $\calU$ has at least two \LQD packets at $\tau_{i+1}$
(otherwise, if there was a queue in $\calU$ with at most one \LQD packet in step $\tau_{i+1}$, no packets would arrive after $\tau_{i+1}$ to this queue by assumption~(A1) %
and consequently, it would not overflow after $\tau_{i+1}$ by Definition~\ref{def:overflow}).
We suppose that $n^i + o^i < M$, since otherwise, the claim holds as queues $\calU$ have in total at most $M$ packets.

Since the \LQD buffer is full in step $\tau_i$,
it is sufficient to observe that any queue $q'$ has at least $s^i_{\LQD}(q') - (\tau_{i+1} - \tau_i)$ packets
in step $\tau_{i+1}$ in the \LQD buffer; more precisely, that packets present in $q'$ at $\tau_i$ are not evicted until step $\tau_{i+1}$.
Suppose for a contradiction that $s^{i+1}_{\LQD}(q') \le s^i_{\LQD}(q') - (\tau_{i+1} - \tau_i) - 1$.
Thus, there must be a step $t\in (\tau_i, \tau_{i+1}]$ such that a packet is evicted from $q'$ at $t$
and $s^t_{\LQD}(q') \le s^i_{\LQD}(q') - (t - \tau_i) - 1$; we take the first such step if there are more.
Note that $t - \tau_i \le s^i_{\max} - 1$ as otherwise, $t - \tau_i \ge s^i_{\max}$ and thus,
$s^t_{\LQD}(q') \le s^i_{\LQD}(q') - (t - \tau_i) - 1 \le s^i_{\max} - s^i_{\max} - 1 = -1$,
which is not possible.

First, suppose that $t - \tau_i = s^i_{\max} - 1$. Then it must be the case that $s^i_{\LQD}(q') = s^i_{\max}$, $s^t_{\LQD}(q') = 0$, and $s^t_{\max} = 1$,
since a packet pending at $\tau_i$ is evicted from $q'$ at $t$. 
This implies that $M$ packets are transmitted by \LQD at $t$. If $t < \tau_{i+1}$, then $n^i + o^i \ge M$, which 
we suppose is not the case. Otherwise, $t = \tau_{i+1}$ and since $s^t_{\max} = 1$, no queue active at $t$ is non-empty for \LQD after $t$ by Assumption~(A1) and thus, $u^i = 0$ and the lemma clearly holds.

Otherwise, we have that $t - \tau_i \le s^i_{\max} - 2$.
As a packet is evicted from $q'$ at $t$, it holds that
\begin{equation}\label{eqn:new active UB simple 1}
s^t_{\LQD}(q') \le s^i_{\max} - (t - \tau_i) - 1\,.
\end{equation}
Recall that there is a queue $\bar{q}$ with $\tau_i = t_{\bar{q}}$, i.e., which overflows for the last time at $\tau_i$.
At $\tau_i$, queue $\bar{q}$ has at least $s^i_{\max} - 1$ packets and as no packet is evicted from $\bar{q}$ after $\tau_i = t_{\bar{q}}$, we get
\begin{equation}\label{eqn:new active UB simple 2}
s^t_{\LQD}(\bar{q})\ge s^i_{\max} - 1 - (t - \tau_i)\,.
\end{equation}
Combining this with~\eqref{eqn:new active UB simple 1}, we obtain $s^t_{\LQD}(\bar{q})\ge s^t_{\LQD}(q')$.
Using $t - \tau_i \le s^i_{\max} - 2$ together with~\eqref{eqn:new active UB simple 2}, we obtain $s^t_{\LQD}(\bar{q})\ge 1$.
As a packet is evicted from $q'$ at $t$, queue $q'$ overflows at $t$,
i.e., it has at least $s^t_{\max} - 1$ many \LQD packets. Thus, $s^t_{\LQD}(\bar{q})\ge s^t_{\max} - 1$
and by Definition~\ref{def:overflow}, $\bar{q}$ overflows at $t$ (here, we also use that the \LQD buffer is full at $t$
as $q'$ overflows and that $s^t_{\LQD}(\bar{q})\ge 1$). However, this contradicts $t_{\bar{q}} = \tau_i < t$.

Hence, any queue $q'$ has at least $s^i_{\LQD}(q') - (\tau_{i+1} - \tau_i)$ \LQD packets in step $\tau_{i+1}$
and the number of \LQD packets in queues $\calU$ in step $\tau_{i+1}$ is at most $n^i + o^i$, which concludes the proof,
as explained above.
\end{proof}

\section{Mapping Transmitted \LQD-extra Packets}
\label{sec:mapping}

So far we derived a lower bound on
$n^i+\alpha\cdot o^i$ for a phase $i$ which depends, among other things, on the number
of queues $|\calL^i_q|$ which are live w.r.t.\ a queue $q$ at time $\tau_i$. To make this bound useful, we now would like to relate $|\calL^i_q|$ to $\hat{e}^i$.

The underlying idea behind establishing a relationship between $|\calL_q^i|$ and $e^i=\sum_{q': \tau_i\ge t_{q'}} e^i_{q'}$ is simple.
By Observation~\ref{obs:t_q},
quantity $e^i$ is bounded by the number of packets that are stored in \OPT's buffer, but not in \LQD's buffer at time $\tau_i$ (more precisely, $e^i \le \sum_{q'} \max\left\{s^i_{\OPT}(q') - s^i_{\LQD}(q'), 0\right\}$). Recall that the \LQD buffer is full in step $\tau_i$. Intuitively, for each packet that \OPT has in its buffer but \LQD has not, there must be a packet that \LQD has in its buffer but \OPT has not. Suppose for a moment that the latter packets are all located in queues in $\calL^i_q$.
Then there can be no more than $\sigma^i_q \cdot |\calL^i_q|$ of them and so, we would have $e^i \le \sigma^i_q \cdot |\calL^i_q|$. Unfortunately, things are more complicated because not all packets of the latter type may be located in live queues. We address this problem by introducing the earlier mentioned careful mapping of transmitted \LQD-extra packets to cancel out some of the packets counted in $e^i$. The following notation will be useful for brevity:
\begin{itemize}
\item{$\calQ^i$:} the set of queues $q$ with $\tau_i\ge t_q$ and $e^i_q > 0$.
  In words, $\calQ^i$ is the set of queues that have overflowed for the last time by step $\tau_i$
  and there are still some \OPT-extra packets to be transmitted from them in phase $i$ or later.
\end{itemize}

To describe our specific mapping, we apply the procedure specified in Algorithm~\ref{alg:mapping} on the solutions
of \LQD and \OPT on the fixed instance $I$.
Intuitively, this procedure maintains a counter $m'_q$ for the number of \LQD-extra packets assigned to a queue $q$,
and assigns \LQD-extra packets in an arbitrary order.
Namely, an \LQD-extra packet transmitted in phase $i$ is mapped to a queue $q'$ in $\calQ^i$ with $m'_q<e^i_q$, i.e.,
with an \OPT-extra packet transmitted from $q$ in phase $i$ or later, such that $t_{q'}$ is the smallest among such queues (breaking ties arbitrarily).
Our values $m_q$ are given as the final values of $m'_q$ after the procedure has been run.
Recall that we use the $m_q$ values to define the $\hat{e}_q$ and $\hat{e}^i_q$ values, which in turn define the scheme to split the \LQD profit.
However, as the mapping procedure is independent from the  $\hat{e}_q$ and $\hat{e}^i_q$ values,
this creates no issue in the analysis. %

\begin{algorithm2e}
\caption{Mapping Procedure}
\label{alg:mapping}
\SetKwComment{Comment}{// }{}
\ForEach{queue $q$}{
  Initialize $m'_q := 0$ \Comment{ counter for packets assigned to $q$}
}
  \ForEach{\LQD-extra packet $p$ transmitted in phase $i$}{
  	\If{there is a queue $q\in\calQ^i$ with $m'_q<e^i_q$}{
	  $\displaystyle{q': = \argmin_{q: q\in\calQ^i\text{ and } m'_q<e^i_q}\{t_q\}}$
      \Comment{breaking ties arbitrarily}
      $m'_{q'} := m'_{q'}+1$ \Comment{assign packet $p$ to queue $q'$}
    }
    \Comment{Otherwise, packet $p$ is not assigned}
  }
\ForEach{queue $q$}{
  $m_q := m'_q$ \Comment{ the final value of $m'_q$}
}
\end{algorithm2e}

We now show a lower bound on the $m_{q'}$ values for a phase $i$. The bound is specific
to a particular queue $q\in\calQ^i$ with $m_q<e_q$; in the following, live and dying queues are w.r.t.\ queue $q$.
For technical reasons, we prove a lower bound on the following quantity:
Let $\overline{m}^i_{q'}$ be the number of \LQD-extra packets transmitted in a phase $j\ge i$ that are mapped to $q'$.
The point is that the constraint $m'_{q'}<e^i_{q'}$ for assigning an \LQD-extra packet to $q'$ in Algorithm~\ref{alg:mapping}
implies that $e^i_{q'}\ge \overline{m}^i_{q'}$, even though it may happen that $e^i_{q'} < m_{q'}$.

For simplicity, let $z^i_q := \sum_{q'\in\calL^i_q}[s^i_{\OPT}(q') > 0]$
be the number of live queues $\calL^i_q$ that are non-empty in \OPT. 
In words, the first term of~\eqref{eqn:mappingBound} below, i.e., %
$\sum_{q'\in\calD^i_q}\max\left\{s^i_{\LQD}(q')-s^i_{\OPT}(q'),0\right\}$, equals the number of packets
that \LQD stores in excess of \OPT in dying queues $q'\in \calD^i_q$ with $s^i_{\LQD}(q') > s^i_{\OPT}(q')$ in step $\tau_i$,
while the second term, i.e., $|\calL^i_q| - z^i_q$, equals 
the number of live queues $\calL^i_q$ that are empty in the \OPT buffer in step $\tau_i$.

\begin{restatable}{lemma}{lemmaMappingBound}\label{lem:mapping_bound}
For any phase $i$ and queue $q\in\calQ^i$ with $m_q < e^i_q$ (i.e., with $\hat{e}^i_q > 0$) we have
\begin{equation}\label{eqn:mappingBound}
\sum_{q'\in\calQ^i} \overline{m}^i_{q'} \ge \sum_{q'\in\calD^i_q}\max\left\{s^i_{\LQD}(q')-s^i_{\OPT}(q'),0\right\} 
								+
                                                                \left(|\calL^i_q|
                                                                  -
                                                                  z^i_q\right)\,.
\end{equation}
\end{restatable}

\begin{proof}
Recall from Definition~\ref{def:liveDying} that
$j_q = \min\{j : \hat{e}^j_q = 0\}$ is the index $j$ of the earliest step $\tau_j$ in which all remaining \OPT-extra packets to be transmitted from $q$ (if any) are canceled out.
Note that $i < j_q$ by the assumption of the lemma and that $m_q < e^j_q$ and $q\in \calQ^j$ for any $j\in [i, j_q)$.

Consider a dying queue $q'\in\calD^i_q$ with $s^i_{\LQD}(q')-s^i_{\OPT}(q') > 0$.
It holds that $e_{q'} = 0$ by Observation~\ref{obs:t_q}.
By Definition~\ref{def:liveDying}, $q'$ will be empty for \LQD in step $\tau_{j_q}$.
Since $q'$ does not overflow after time $\tau_i$ (and hence \LQD will accept all packets which may arrive to $q'$ after time $\tau_i$), at least $s^i_{\LQD}(q')-s^i_{\OPT}(q')$ \LQD-extra packets are transmitted from $q'$ from time $\tau_i$ until time $\tau_{j_q} - 1$, i.e., in phases $j\in [i,j_q)$.
Using $m_q < e^j_q$ and $q\in \calQ^j$ for any $j\in [i, j_q)$,
all these \LQD-extra packets are allocated to queues $\overline{q}$
that satisfy $t_{\overline{q}} \le t_q \le \tau_i$ and $e^i_{\overline{q}} > 0$; the second property holds
as $t_{\overline{q}} \le t_q\le \tau_i$ and as $\overline{q}\in \calQ^j$ for a phase $i\le j < j_q$
in which the \LQD-extra packet assigned to it is transmitted, by Algorithm~\ref{alg:mapping}. 
Thus, such queues $\overline{q}$ are part of the set $\calQ^i$.

In addition, there are $|\calL^i_q| - z^i_q$ live queues in step $\tau_i$ that are empty in the \OPT buffer.
Hence, \LQD transmits $|\calL^i_q| - z^i_q$ \LQD-extra packets in step $\tau_i$ from such queues,
and as $m_q < e^i_q$, all these packets are assigned to queues $\overline{q}\in \calQ^i$ by Algorithm~\ref{alg:mapping}.
\end{proof}

Finally, we bound the number of \OPT-extra packets which are not canceled out by \LQD-extra packets.
Equivalently, for a queue $q$, we show a lower bound on $|\calL^i_q|$ in terms of $\hat{e}^i$.
The lemma below in particular implies that if $\hat{e}^i_q > 0$ then also $|\calL^i_q| \ge 1$, i.e., there is at least 
one live queue w.r.t.\ queue $q$.

\begin{restatable}{lemma}{lemmaUnmatchedExtraBound}\label{lem:unmatchedExtraBound}
For any phase $i$ and queue $q\in\calQ^i$ with $\hat{e}^i_q > 0$, we have that
$\displaystyle
\hat{e}^i \le (\sigma^i_q - 1)\cdot |\calL^i_q|\,.
$
\end{restatable}

\begin{proof}
First note that
\begin{equation*}%
\hat{e}^i = \sum_{q'\in \calQ^i} \hat{e}^i_{q'} = \sum_{q'\in \calQ^i} \max\{e^i_{q'} - m_{q'}, 0\}
\le \sum_{q'\in \calQ^i} \max\{e^i_{q'} - \overline{m}^i_{q'}, 0\}
= \sum_{q'\in \calQ^i} (e^i_{q'} - \overline{m}^i_{q'})\,,
\end{equation*}
where the inequality holds by $\overline{m}^i_{q'}\le m_{q'}$
and the last step follows from $e^i_{q'}\ge \overline{m}^i_{q'}$, by the definition of $\overline{m}^i_{q'}$ and Algorithm~\ref{alg:mapping}.
Using Lemma~\ref{lem:mapping_bound}, we obtain
\begin{equation}\label{eqn:unmatchedExtraBound 2}
\hat{e}^i \le \sum_{q'\in \calQ^i} (e^i_{q'}) - \sum_{q'\in\calD^i_q}\max\left\{s^i_{\LQD}(q')-s^i_{\OPT}(q'),0\right\} 
								- \left(|\calL^i_q| - z^i_q\right)
\end{equation}
By the definition of $\sigma^i_q$ in~\eqref{eqn:sigma_def} and since the \LQD buffer is full in step $\tau_i$, we have
\begin{equation}\label{eqn:unmatchedExtraBound-LQD_buffer}
M= \sigma^i_q \cdot |\calL^i_q| + \sum_{q'\in \calD^i_q} s^i_{\LQD}(q')\,.
\end{equation}
Regarding the \OPT buffer, we have 
\begin{align}
M\ge \sum_{q'} s^i_{\OPT}(q')
&\ge \sum_{q'\in \calQ^i} \max\left\{s^i_{\OPT}(q') - s^i_{\LQD}(q'), 0\right\} + \sum_{q'} \min\{s^i_{\LQD}(q'),s^i_{\OPT}(q')\}
\nonumber\\
&\ge \sum_{q'\in \calQ^i} \left(e^i_{q'}\right) + \sum_{q'} \min\{s^i_{\LQD}(q'),s^i_{\OPT}(q')\}
\nonumber\\
&\ge \sum_{q'\in \calQ^i} \left(e^i_{q'}\right) + \sum_{q'\in \calD^i_q} \min\{s^i_{\LQD}(q'),s^i_{\OPT}(q')\}
	 + z^i_q
\,, \label{eqn:unmatchedExtraBound 3}
\end{align}
where the inequality in the second line uses Observation~\ref{obs:t_q}.
Combining~\eqref{eqn:unmatchedExtraBound-LQD_buffer} and~\eqref{eqn:unmatchedExtraBound 3}, we obtain
\begin{equation*}%
\sum_{q'\in \calQ^i} \left(e^i_{q'}\right) + \sum_{q'\in \calD^i_q} \min\{s^i_{\LQD}(q'),s^i_{\OPT}(q')\} + z^i_q
\le \sigma^i_q \cdot |\calL^i_q| + \sum_{q'\in \calD^i_q} s^i_{\LQD}(q')
\end{equation*}
After rearranging, we get
\begin{equation}\label{eqn:unmatchedExtraBound 5}
\sum_{q'\in \calQ^i} \left(e^i_{q'}\right) - \sum_{q'\in\calD^i_q}\max\left\{s^i_{\LQD}(q')-s^i_{\OPT}(q'),0\right\} + z^i_q
\le \sigma^i_q \cdot |\calL^i_q|\,.
\end{equation}
Finally, plugging~\eqref{eqn:unmatchedExtraBound 5} into Equation~\eqref{eqn:unmatchedExtraBound 2},
we get $\hat{e}^i \le \sigma^i_q \cdot |\calL^i_q| - |\calL^i_q|$, as desired.
\end{proof}

\section{Putting It All Together}
\label{sec:puttingItAllTogether}
\pvchanged{
In this section, we complete the proof of $1.6918$-competitiveness for \LQD.
The first part is to obtain suitable lower bounds on the profit assigned to a queue in a phase.
Next, we sum these lower bounds over all phases and derive a lower bound for this sum.
Finally, we optimize the parameter $\alpha$ to maximize $\rho$ (and thus, minimize the competitive ratio upper bound)
subject to $\Phi_q \ge \rho\cdot \hat{e}_q$ for any queue $q$.

\paragraph{L-increase and S-increase.}
Before proving the lower bounds, it is convenient to split the quantity $\Delta^{i}\Phi_q$ into two quantities that are analyzed separately.
Apart from $\alpha$, we use another parameter $\beta\in (0, 1)$ such that $\alpha + \beta < 1$; namely we will set $\beta = 1-\sqrt{1-\alpha}-\alpha/2$ (we will require that $\alpha$ is not too close to $1$ so that $\alpha + \beta < 1$; namely, this setting of $\beta$ requires that $\alpha < 2(\sqrt{2}-1) \approx 0.828$).
Given such $\alpha$ and $\beta$, we adjust the definition of $\Delta^{i}\Phi_q$ in~\eqref{eqn:splitting profit def init} to
\begin{equation}\label{eqn:splitting profit def}
	\Delta^{i}\Phi_q = \underbrace{\frac{\hat{e}^i_q}{\hat{e}^i} \cdot (n^{i} + \alpha \cdot o^{i} - \Delta^i \Psi)
		+ \beta\cdot o^i_q}_\text{L-increase} \,\, + \,\, \underbrace{(1 - \alpha - \beta)\cdot o^i_q}_\text{S-increase}
\end{equation}

We call the first two terms in Equation~\eqref{eqn:splitting profit def}
(i.e., $(\hat{e}^i_q / \hat{e}^i) \cdot (n^i + \alpha \cdot o^i - \Delta^i \Psi) + \beta\cdot o^i_q$)
the \emph{L-increase} for $q$ as they will be mainly useful for ``long'' queues (with relatively high $\hat{e}_q$).
We call the last term, $(1 - \alpha - \beta)\cdot o^i_q$, the \emph{S-increase} for $q$
as it works well for ``short'' queues.

Analyzing the S-increases is relatively easy. Most of the remainder of this section is thus focused on showing lower bounds on the L-increase of a queue $q$ in phase $i$, using the lemmas developed in previous sections.
This will be divided into two cases, according to whether the value of $\sigma^i_q$ decreases or not (w.r.t.\ variable $i$). %
} %

\subsection{Lower Bounds on the L-Increase}
\label{sec:LB on L-increase}

In this section, for a queue $q$, we show lower bounds on the L-increase for a phase $i$ with $\tau_i\ge t_q$ and $\hat{e}^i_q>0$.

\bigskip

We start with a generic lower bound on $n^i + \alpha \cdot o^i - \Delta^i \Psi$.
\begin{claim}\label{clm:lower bound on first part}
For any phase $i$ and queue $\tau_i\ge t_q$ with $\hat{e}^i_q > 0$, we have that
    \begin{align*}
        n^i + \alpha \cdot o^i - \Delta^i \Psi
         & \ge \frac{\hat{e}^i}{\sigma^i_q-1} \cdot \Big(\tau_{i+1} - \tau_i - \alpha\cdot \min\Big\{\frac{\tau_{i+1}-\tau_i}{2}, \frac{\sigma^i_q}{\sigma^{i+1}_q}-1, \sigma^i_q-\sigma^{i+1}_q\Big\}\Big) \,.
    \end{align*}
\end{claim}
\begin{proof}
    We start by deriving several upper bounds on $u^i -d^i_q - v^i$. The first two bounds follow from Lemma~\ref{lem:dying_bound}:
    \begin{equation}\label{eqn:LB1stpart-eqn1}
        u^i -d^i_q - v^i \le u^i - \frac{d^i_q + v^i}{\sigma^{i+1}_q} \le \frac{\sigma^i_q-\sigma^{i+1}_q}{\sigma^{i+1}_q}\cdot|\calL^i_q|
    \end{equation}
    and
    \begin{equation}\label{eqn:LB1stpart-eqn2}
        u^i -d^i_q - v^i \le  \sigma_q^{i+1} \cdot u^i -d^i_q - v^i \le (\sigma_q^i-\sigma_q^{i+1}) \cdot |\calL^i_q| \,.
        \end{equation}
    More precisely, the first inequality in~\eqref{eqn:LB1stpart-eqn1}
    uses $\sigma_q^{i+1}\ge 1$ and the second one holds by Lemma~\ref{lem:dying_bound}; \eqref{eqn:LB1stpart-eqn2} is derived similarly.
    We therefore have
    \begin{align*}\label{eqn:udv-upper-one}
        u^i -d^i_q - v^i \le \min\Big\{\frac{\sigma^i_q}{\sigma^{i+1}_q}-1,\sigma^i_q-\sigma^{i+1}_q\Big\}\cdot |\calL^i_q| \,.
    \end{align*}
    Using Observation~\ref{obs:live_and_dying}, we obtain a lower bound on $n^i + \alpha \cdot o^i - \Delta^i \Psi$:
    \begin{equation}
        \begin{aligned}\label{eqn:profit first two lower bounds}
            n^i + \alpha \cdot o^i - \Delta^i \Psi & \ge |\calL^i_q|\cdot (\tau_{i+1}-\tau_i) + \alpha \cdot (d^i_q - u^i + v^i)                                                                      \\
                                                   & \ge |\calL^i_q|\cdot \Big(\tau_{i+1}-\tau_i - \alpha \cdot\min\Big\{\frac{\sigma^i_q}{\sigma^{i+1}_q}-1,\sigma^i_q-\sigma^{i+1}_q\Big\}\Big) \,.
        \end{aligned}
    \end{equation}
    Furthermore, Lemma~\ref{lem:new active UB simple} gives
    \[
        u^i -v^i \le u^i \le \frac12\cdot (n^i + o^i) \,.
    \]
    Using this together with Observation~\ref{obs:live_and_dying} we get
    \begin{align}\label{eqn:profit third lower bound}
        n^i + \alpha \cdot o^i - \Delta^i \Psi \ge n^i \cdot \left(1-\frac{\alpha}{2}\right) + o^i\cdot \frac{\alpha}{2} \ge n^i \cdot \left(1-\frac{\alpha}{2}\right) \ge |\calL^i_q|\cdot \left(\tau_{i+1}-\tau_i - \alpha\cdot \frac{\tau_{i+1}-\tau_i}{2}\right) \,.
    \end{align}
    Combining the upper bounds from \eqref{eqn:profit first two lower bounds} and \eqref{eqn:profit third lower bound} and using $|\calL^i_q| \ge \hat{e}^i/(\sigma^i_q-1)$ from Lemma~\ref{lem:unmatchedExtraBound} implies the claim.
\end{proof}

The following lemma is mainly useful for the case $\sigma^{i+1}_q \ge \sigma^i_q$ and follows directly from Claim~\ref{clm:lower bound on first part}
(using only the first and second terms in the minimum expression and then multiplying the inequality by $\hat{e}^i_q/\hat{e}^i$). %

\begin{lemma}\label{lem:LB on L-increase}
  Consider any queue $q$ with $\hat{e}_q>0$ and any phase $i$ with $\tau_i\ge t_q$ and $\hat{e}^i_q>0$.
  Then the L-increase in phase $i$ for queue $q$ satisfies
  \begin{equation*}\label{eqn:LB on L-increase}
    \frac{\hat{e}^i_q}{\hat{e}^i} \cdot (n^i + \alpha \cdot o^i - \Delta^i \Psi)
    \ge \hat{e}^i_q \cdot \frac{\alpha \cdot \max\left\{\sigma^{i+1}_q - \sigma^i_q,1-\sigma^i_q/\sigma^{i+1}_q\right\} + (\tau_{i+1} - \tau_i)}{\sigma^i_q - 1}\,.
  \end{equation*}
\end{lemma}

Next, we deal with the (more involved) case when $\sigma^{i+1}_q < \sigma^i_q$.
Note that when $\hat{e}^{i+1}_q>0$, we have that $|\calL^{i+1}_q|\ge 1$ by Lemma~\ref{lem:unmatchedExtraBound}
and thus $\sigma^{i+1}_q\ge 2$ by Observation~\ref{obs:sigma_lower_bound},
while if $\hat{e}^{i+1}_q = 0$, the last term in~\eqref{eqn:L-charge per phase LB - decreasing sigma} below is defined to be 0, even if $\sigma^{i+1}_q = 1$.

\begin{lemma}\label{lem:L-charge LB - decreasing sigma}
Consider any queue $q$ with $\hat{e}_q>0$ and any phase $i$ with $\tau_i\ge t_q$ and $\hat{e}^i_q>0$, and suppose that $\sigma^{i+1}_q < \sigma^i_q$.
Then, for $\beta \ge 1-\sqrt{1-\alpha}-\alpha/2$,
it holds that
\begin{equation}\label{eqn:L-charge per phase LB - decreasing sigma}
\frac{\hat{e}^i_q}{\hat{e}^i} \cdot (n^i + \alpha \cdot o^i - \Delta^i \Psi) + \beta\cdot o^i_q
\ge\sum_{t = \tau_i}^{\tau_{i+1} - 1} \left(\frac{\hat{e}^t_q}{\sigma^i_q - 1}\right) - \frac{g^i_q}{2(\sigma^i_q - 1)}
	- \hat{e}^{i+1}_q\cdot \left(\frac{1}{\sigma^{i+1}_q - 1} - \frac{1}{\sigma^i_q - 1}\right)\,,
\end{equation}
where $g^i_q$ is the number of steps $t\in [\tau_i, \tau_{i+1})$ with $\hat{e}^t_q > 0$ and $s^t_\LQD(q) = 0$.
\end{lemma}

Before we prove this lemma, we start by establishing a lower bound on $\beta \cdot o^i_q$.
\begin{claim}\label{clm:lower bound on second part}
	For $\beta \ge 1-\sqrt{1-\alpha}-\alpha/2$,
	\begin{equation}\label{eqn:lower bound on second part}
			\beta\cdot o^i_q \ge  \frac{g^i_q}{\sigma^i_q-1}\cdot\Big(\alpha\cdot \min\Big\{\frac{\tau_{i+1}-\tau_i}{2}, \frac{\sigma^i_q}{\sigma^{i+1}_q}-1, \frac{\sigma^i_q-1}{2} \Big\} - \frac{g^i_q}{2} - h^i_q\Big) \,,
	\end{equation}
	where $h^i_q$ is the number of steps $t$ in the $i$-th phase such that $\hat{e}^t_q = 0$ (i.e., all 
	\OPT-extra packets are already transmitted or canceled out).
\end{claim}
\begin{proof}

If the right-hand side is smaller than 0, there is nothing to show. Therefore, in the following we assume that the right-hand side is non-negative.

First note that from the fixed queue $q$, during phase $i$, first both \LQD and \OPT transmit $o^i_q$ packets (possibly $o^i_q = 0$)
and then $q$ becomes empty for \LQD, so the adversary transmits \OPT-extra packets,
until $q$ becomes inactive or until the end of the phase.
Observe that $\tau_{i+1} - \tau_i = o^i_q + g^i_q + h^i_q$ and that $\hat{e}^{i+1}_q = \hat{e}^i_q - g^i_q$. 
The right-hand side of~\eqref{eqn:lower bound on second part} is upper bounded by
\begin{align}
\frac{g^i_q}{\sigma^i_q-1}\cdot\Big(\alpha\cdot \min\Big\{\frac{\tau_{i+1}-\tau_i}{2}, \frac{\sigma^i_q-1}{2} \Big\} - \frac{g^i_q}{2} - h^i_q\Big)
&\le
\frac{g^i_q}{\tau_{i+1}-\tau_i}\cdot\Big(\alpha\cdot \frac{\tau_{i+1}-\tau_i}{2} - \frac{g^i_q}{2} - h^i_q\Big)
\label{eqn:clm:LB on 2nd part - step 2}\\
&=
\frac{g^i_q}{o^i_q+g^i_q+h^i_q}\cdot\Big(\alpha\cdot \frac{o^i_q+g^i_q+h^i_q}{2} - \frac{g^i_q}{2} - h^i_q\Big)
\label{eqn:clm:LB on 2nd part - step 3}\\
&=
\frac{g^i_q}{o^i_q+g^i_q+h^i_q}\cdot \frac{\alpha\cdot o^i_q - (1-\alpha) \cdot g^i_q - (2-\alpha) \cdot h^i_q}{2}
\nonumber\\
&\le
\frac{g^i_q}{o^i_q+g^i_q}\cdot \frac{\alpha\cdot o^i_q - (1-\alpha) \cdot g^i_q}{2}
\label{eqn:clm:LB on 2nd part - step 5}\\
&\le
o^i_q\cdot \Big(1-\sqrt{1-\alpha}-\frac{\alpha}{2}\Big)
\le o^i_q\cdot \beta \,,
\label{eqn:clm:LB on 2nd part - step 6}
\end{align}
where \eqref{eqn:clm:LB on 2nd part - step 2} holds by a case analysis, namely, if $\tau_{i+1}-\tau_i \le \sigma^i_q-1$,
then we use $g^i_q / (\sigma^i_q-1)\le g^i_q / (\tau_{i+1}-\tau_i)$, while in the other case, $\min\{\tau_{i+1}-\tau_i, \sigma^i_q-1 \} = \sigma^i_q-1$ and we apply 
$(- g^i_q/2 - h^i_q) / (\sigma^i_q-1) \le (- \frac{g^i_q}{2} - h^i_q) / (\tau_{i+1}-\tau_i)$;
inequality~\eqref{eqn:clm:LB on 2nd part - step 3} follows from $\tau_{i+1} - \tau_i = o^i_q + g^i_q + h^i_q$; and inequality~\eqref{eqn:clm:LB on 2nd part - step 6} follows because expression~\eqref{eqn:clm:LB on 2nd part - step 5} is maximized for $g^i_q = o^i_q \cdot (1/\sqrt{1-\alpha}-1)$ (in more detail, the inequality clearly holds for $o^i_q = 0$ and otherwise, letting $x = g^i_q / o^i_q\ge 0$, expression~\eqref{eqn:clm:LB on 2nd part - step 5} divided by $o^i_q$ equals $\frac{x}{2+2x}\cdot (\alpha - (1 - \alpha)\cdot x)$, which, for a fixed $0 < \alpha < 1$, is maximized for $x = \Big(1-\sqrt{1-\alpha}-\frac{\alpha}{2}\Big)$ by a routine calculation).
\end{proof}

With this, we are ready to prove Lemma~\ref{lem:L-charge LB - decreasing sigma} by combining Claims~\ref{clm:lower bound on first part} and \ref{clm:lower bound on second part}.

\begin{proof}[Proof of Lemma~\ref{lem:L-charge LB - decreasing sigma}]
\begin{align}
	\frac{\hat{e}^i_q}{\hat{e}^i} &\cdot (n^i + \alpha \cdot o^i - \Delta^i \Psi) + \beta\cdot o^i_q
	\nonumber\\
	&\ge 
	\frac{\hat{e}^i_q}{\sigma^i_q-1} \cdot \Big(\tau_{i+1}-\tau_i - \alpha\cdot \min\Big\{\frac{\tau_{i+1}-\tau_i}{2}, \frac{\sigma}{\sigma'}-1,\frac{\sigma^i_q-1}{2}\Big\}\Big) + \beta\cdot o^i_q
	\label{eqn:lem:L-charge LB - decreasing sigma - step1}\\
	&\ge
	\frac{\hat{e}^i_q}{\sigma^i_q-1}\cdot (\tau_{i+1}-\tau_i)- \alpha\cdot \frac{\hat{e}^i_q-g^i_q}{\sigma^i_q-1} \cdot \min\Big\{\frac{\tau_{i+1}-\tau_i}{2}, \frac{\sigma^i_q}{\sigma^{i+1}_q}-1,\frac{\sigma^i_q-1}{2}\Big\} - \frac{(g^i_q)^2}{2\cdot(\sigma^i_q-1)} - \frac{g^i_q\cdot h^i_q}{\sigma^i_q-1}
	\label{eqn:lem:L-charge LB - decreasing sigma - step2}\\
&\ge
\sum_{t = \tau}^{\tau' - 1} \frac{\hat{e}^t_q}{\sigma^i_q - 1} - \alpha\cdot \frac{\hat{e}^i_q-g^i_q}{\sigma^i_q-1} \cdot \min\Big\{\frac{\tau_{i+1}-\tau_i}{2}, \frac{\sigma^i_q}{\sigma^{i+1}_q}-1,\frac{\sigma^i_q-1}{2}\Big\} - \frac{g^i_q}{2\cdot(\sigma^i_q-1)}
\label{eqn:lem:L-charge LB - decreasing sigma - step3}\\
	&=
\sum_{t = \tau}^{\tau' - 1} \frac{\hat{e}^t_q}{\sigma^i_q - 1} - \alpha\cdot \frac{\hat{e}^{i+1}_q}{\sigma^i_q-1} \cdot \min\Big\{\frac{\tau_{i+1}-\tau_i}{2}, \frac{\sigma^i_q}{\sigma^{i+1}_q}-1,\frac{\sigma^i_q-1}{2}\Big\} - \frac{g^i_q}{2\cdot(\sigma^i_q-1)}
\label{eqn:lem:L-charge LB - decreasing sigma - step4}\\
	&\ge
\sum_{t = \tau}^{\tau' - 1} \frac{\hat{e}^t_q}{\sigma^i_q - 1} - \alpha\cdot \frac{\hat{e}^{i+1}_q}{\sigma^i_q-1} \cdot \Big(\frac{\sigma^i_q}{\sigma^{i+1}_q}-1\Big) - \frac{g^i_q}{2\cdot(\sigma^i_q-1)}
\nonumber\\
	&\ge
\sum_{t = \tau}^{\tau' - 1} \frac{\hat{e}^t_q}{\sigma^i_q - 1} - \alpha\cdot \frac{\hat{e}^{i+1}_q}{\sigma^i_q-1} \cdot \Big(\frac{\sigma^i_q-1}{\sigma^{i+1}_q-1}-1\Big) - \frac{g^i_q}{2\cdot(\sigma^i_q-1)}
\label{eqn:lem:L-charge LB - decreasing sigma - step6}\\
	&=
\sum_{t = \tau}^{\tau' - 1} \frac{\hat{e}^t_q}{\sigma^i_q - 1} - \alpha\cdot \hat{e}^{i+1}_q \cdot \Big(\frac{1}{\sigma^{i+1}_q-1}-\frac{1}{\sigma^i_q-1}\Big) - \frac{g^i_q}{2\cdot(\sigma^i_q-1)}\,, \nonumber
\end{align}
where inequality~\eqref{eqn:lem:L-charge LB - decreasing sigma - step1} uses Claim~\ref{clm:lower bound on first part}, the step in~\eqref{eqn:lem:L-charge LB - decreasing sigma - step2} uses Claim~\ref{clm:lower bound on second part}, inequality~\eqref{eqn:lem:L-charge LB - decreasing sigma - step3} uses
\begin{equation*}
	\sum_{t = \tau_i}^{\tau_{i+1} - 1} \frac{\hat{e}^t_q}{\sigma^i_q - 1}
	= \frac{\hat{e}^i_q\cdot (\tau_{i+1} - \tau_i)}{\sigma^i_q - 1} - \frac{g^i_q\cdot (g^i_q - 1)}{2(\sigma^i_q - 1)}
	- \frac{\hat{e}^i_q\cdot h^i_q}{\sigma^i_q - 1}\,,
\end{equation*}
the fourth step (inequality~\eqref{eqn:lem:L-charge LB - decreasing sigma - step4}) uses $\hat{e}^{i+1}_q = \hat{e}^i_q - g^i_q$, and inequality~\eqref{eqn:lem:L-charge LB - decreasing sigma - step6} follows because $\sigma^{i+1}_q < \sigma^i_q$.
The lemma follows since $\alpha<1$.
\end{proof}

\subsection{Total \LQD Profit Assigned to a Queue}
\label{sec:summing L-increases}

Fix a queue $q$ with $\hat{e}_q > 0$, i.e., with transmitted \OPT-extra packets that are not canceled out
by transmitted \LQD-extra packets.
We now show a lower bound on $\sum_i \Delta^i \Phi_q$. Recall that $\Delta^i \Phi_q > 0$ only for phases $i$
with $\tau_i\ge t_q$ and $\hat{e}^i_q > 0$.

First, we bound the sum of S-increases. Note that as $q$ overflows at $t_q$, \LQD stores for this queue at least
$s^{t_q}_{\max} - 1$ packets in step $t_q$, and since it 
does not overflow after $t_q$, \LQD transmits at least $s^{t_q}_{\max} - 1\ge \left\lceil\sigma^{t_q}_q\right\rceil - 1$
packets from $q$ at or after $t_q$,
where the inequality is from Observation~\ref{obs:sigma_upper_bound} (recall that $t_q = \tau_i$ for some phase $i$).
Thus, the sum of S-increases is at least $(1 - \alpha - \beta)\cdot \left(\left\lceil\sigma^{t_q}_q\right\rceil - 1\right)$.
The next lemma shows a bound on the total L-increase assigned to queue $q$.

\begin{lemma}\label{lem:L-increases sum}
Assuming $\alpha \le 2/3$ and using $b_0 = \left\lceil\sigma^{t_q}_q\right\rceil - 1$ for simplicity,
the sum of L-increases assigned to a queue $q$ with $\hat{e}_q > 0$ over all phases is at least
\begin{equation*}
 \alpha\cdot\hat{e}_q\cdot \left(1 + \frac{b_0}{\hat{e}_q}\right)\cdot \ln\left(1 + \frac{\hat{e}_q}{b_0}\right)+ \alpha\cdot\hat{e}_q\cdot \ln\left(\frac1\alpha\right)\,.
\end{equation*}
\end{lemma}

\begin{proof}
Consider phase $i$ with $\tau_i\ge t_q$ and $\hat{e}^i_q>0$.
As in the previous section, we have two cases.
If $\sigma^{i+1}_q \ge \sigma^i_q$, then using Lemma~\ref{lem:LB on L-increase},
in phase $i$ queue $q$ receives L-increase of at least
\begin{equation}\label{eqn:Lch incr sigma}
	\hat{e}^i_q \cdot \frac{\alpha \cdot \left(\sigma^{i+1}_q - \sigma^i_q\right) + (\tau_{i+1} - \tau_i)}{\sigma^i_q - 1}
	\ge \sum_{t = \tau_i}^{\tau_{i+1} - 1}\left(\frac{\hat{e}^t_q}{\sigma^i_q - 1}\right)
		 + \hat{e}^{\tau_{i + 1} - 1}_q \cdot \frac{\alpha \cdot \left(\sigma^{i+1}_q - \sigma^i_q\right)}{\sigma^i_q - 1}\,.
\end{equation}
Otherwise, $\sigma^{i+1}_q < \sigma^i_q$, and Lemma~\ref{lem:L-charge LB - decreasing sigma} shows that the total L-increase during the $i$-th phase is at least
\begin{equation}\label{eqn:Lch decr sigma}
\sum_{t = \tau_i}^{\tau_{i+1} - 1} \left(\frac{\hat{e}^t_q}{\sigma^i_q - 1}\right) - \frac{g^i_q}{2(\sigma^i_q - 1)}
	- \hat{e}^{i+1}_q\cdot \left(\frac{1}{\sigma^{i+1}_q - 1} - \frac{1}{\sigma^i_q - 1}\right)\,,
\end{equation}
where $g^i_q$ is the number of steps $t\in [\tau_i, \tau_{i+1})$ with $\hat{e}^t_q > 0$ and $s^t_\LQD(q) = 0$.

We now simplify the sum of the right-hand sides of~\eqref{eqn:Lch incr sigma} or~\eqref{eqn:Lch decr sigma} over all phases.
Let $i_q$ be such that $\tau_{i_q} = t_q$; index $i_q$ is well-defined by the definition of phases. %
We define a non-decreasing sequence $a_{i_q}, a_{i_q + 1}, a_{i_q + 2}, \dots$
as $a_{i_q} = \left\lceil\sigma^{i_q}_q\right\rceil - 1$ and for $i > i_q$, we set
$a_i = \max\{\sigma^i_q - 1, a_{i-1}\}$
(note that when $\calL^i_q = \emptyset$ and thus $\sigma^i_q = 1$, we have $a_i = a_{i-1}$).
By definition, the sequence of $a_i$'s upper bounds sequence $\{\sigma^i_q - 1\}_{i \ge i_q}$.

Note that $a_{i+1} - a_i \le \max\{\sigma^{i+1}_q - \sigma^i_q, 0\}$, since $a_i\ge \sigma^i_q - 1$ and
$a_{i+1}$ is equal either to $a_i$, or to $\sigma^{i+1}_q - 1$.
This in particular implies that $a_{i+1} > a_i$ only if $\sigma^{i+1}_q > \sigma^i_q$.
Further, observe that if $\hat{e}^{i+1}_q > 0$, then the last term of~\eqref{eqn:Lch decr sigma} plus the first term
of the appropriate inequality for phase $i+1$ becomes
\begin{equation*}
- \hat{e}^{i+1}_q\cdot \left(\frac{1}{\sigma^{i+1}_q - 1} - \frac{1}{\sigma^i_q - 1}\right) + \frac{\hat{e}^{i+1}_q}{\sigma^{i+1}_q - 1}
= \frac{\hat{e}^{i+1}_q}{\sigma^i_q - 1} \ge \frac{\hat{e}^{i+1}_q}{a_{i+1}}\,,
\end{equation*}
where we use that $a_{i+1}\ge a_i\ge \sigma^i_q - 1$.
Using these two observations, we lower bound the sum of the right-hand side of~\eqref{eqn:Lch incr sigma} or~\eqref{eqn:Lch decr sigma} over all phases by
\begin{equation*}
\sum_{i\ge i_q} \left( \sum_{t = \tau_i}^{\tau_{i+1} - 1} \left(\frac{\hat{e}^t_q}{a_i}\right) + 
	\hat{e}^{\tau_{i + 1} - 1}_q \cdot \frac{\alpha \cdot (a_{i+1} - a_i)}{a_i}
	- \frac{g^i_q}{2a_i}\right) \,.
\end{equation*}

Next, we switch from the summation over phases to a sum over steps. We define another sequence $b_0, b_1, \dots$
as $b_j = a_i$, where $i$ is the index of the phase which contains step $t_q + j$, i.e., satisfying $t_q + j\in [\tau_i, \tau_{i+1})$.
Note that the sequence of $b_j$'s is also non-decreasing and that we index $b_j$'s from $0$ instead of from $t_q$,
which will be convenient below. Thus, we obtain a lower bound of
\begin{equation}\label{eqn:Lch sum 1.5}
\sum_{t\ge t_q} \left(\hat{e}^t_q\cdot \frac{\alpha \cdot (b_{t - t_q + 1} - b_{t - t_q}) + 1}{b_{t - t_q}}
	- \frac{g^t_q}{2b_{t - t_q}}\right)\,,
\end{equation}
where $g^t_q = 1$ if $\hat{e}^t_q > 0$ and $s^t_\LQD(q) = 0$ (i.e., if an \OPT-extra packet is transmitted from $q$ at $t$ and is not canceled out),
and $g^t_q = 0$ otherwise.
Recall that as long as $q$ is non-empty for \LQD, $\hat{e}^t_q$ does not change (and $g^t_q = 0$) and then decreases by one in each step (and $g^t_q = 1$) until.
Since \LQD transmits at least $\left\lceil\sigma^{t_q}_q\right\rceil - 1 = b_0$ packets from $q$ in steps $t\ge t_q$,
it follows that $\hat{e}^t_q = \hat{e}_q$ for $t\in [t_q, t_q + b_0 - 1]$.
Below, instead of summing over steps $t\ge t_q$, we sum over $j\ge 0$.
Thus, we can bound~\eqref{eqn:Lch sum 1.5} from below by
\begin{equation}\label{eqn:Lch sum 2}
\sum_{j = 0}^{b_0 - 1} \frac{\alpha\cdot (b_{j+1} - b_j) + 1}{b_j}\cdot \hat{e}_q
+
\sum_{j = b_0}^{b_0 + \hat{e}_q - 1} \left( \frac{\alpha\cdot (b_{j+1} - b_j) + 1}{b_j}\cdot (\hat{e}_q +b_0 - j) 
- \frac{1}{2b_j} \right)\,.
\end{equation}
It remains to lower bound~\eqref{eqn:Lch sum 2}. %
We remark that the bound below holds for any positive integers $b_0$ and $\hat{e}_q$ and any non-decreasing sequence
of positive numbers $b_0, b_1, b_2, \dots$

Rewriting~\eqref{eqn:Lch sum 2} so that instead of summing over steps, we sum over the $\hat{e}_q$ transmitted \OPT-extra packets that are not canceled out,
we obtain
\begin{equation}\label{eqn:Lch-sum-bound1_2 w/ alpha and minus term}
\sum_{f = 1}^{\hat{e}_q} \left( \sum_{j = 0}^{b_0 + f-1} \left(\frac{\alpha\cdot (b_{j+1} - b_j) + 1}{b_j} \right)
		- \frac{1}{2b_{b_0 + f - 1}} \right)\,.
\end{equation}
We first give a useful lower bound for each summand of the outer sum.

\begin{claim}\label{clm:inner-sum-bound w/ alpha and minus term}
Assuming $\alpha \le 2/3$,
for any integer $m\ge b_0 + 1$ and  any non-decreasing sequence of positive numbers $b_0,b_1,b_2, \dots$, 
we have
\[
\sum_{j=0}^{m-1} \frac{\alpha\cdot (b_{j+1}-b_j) +1}{b_j}  - \frac{1}{2b_{m-1}}
\ge \alpha\cdot\left(  1 + \ln\left(\frac{1}{\alpha}\right) - \ln(b_0) + \ln(m)  \right) \,.
\]
\end{claim}
\begin{proof}
\begin{align}
\sum_{j=0}^{m-1} \frac{\alpha\cdot (b_{j+1}-b_j) +1}{b_j} - \frac{1}{2b_{m-1}}
&\ge \alpha\cdot \sum_{j=0}^{m-1} \frac{b_{j+1}-b_j}{b_j}+\frac{m - \frac{3}{2}}{b_m} + \frac{1}{b_0}
\nonumber\\
&\ge \alpha\cdot \sum_{j=0}^{m-1} \ln\Big(\frac{b_{j+1}}{b_j}\Big)+\frac{m - \frac{3}{2}}{b_m} + \frac{1}{b_0}
\label{eqn:inner-sum-bound w/ alpha and minus term - step2}\\
& = \alpha\cdot \ln\Big(\frac{b_{m}}{b_0}\Big)+\frac{m - \frac{3}{2}}{b_m} + \frac{1}{b_0}
\nonumber\\
&\ge \alpha + \alpha\cdot  \ln\left(\frac{m - \frac{3}{2}}{\alpha\cdot b_0}\right) + \frac{1}{b_0}
\label{eqn:inner-sum-bound w/ alpha and minus term - step4}\\
&= \alpha + \alpha\cdot \ln\left(\frac{1}{\alpha}\right) - \alpha\cdot \ln(b_0) + \alpha\cdot \ln\left(m - \frac{3}{2}\right) + \frac{1}{b_0}
\nonumber\\
&\ge \alpha + \alpha\cdot \ln\left(\frac{1}{\alpha}\right) - \alpha\cdot \ln(b_0) + \alpha\cdot \ln(m)\,,
\label{eqn:inner-sum-bound w/ alpha and minus term - step6}
\end{align}
where inequality~\eqref{eqn:inner-sum-bound w/ alpha and minus term - step2} follows from $x-1\ge \ln(x)$ for $x>0$,
step \eqref{eqn:inner-sum-bound w/ alpha and minus term - step2} uses that the expression is minimized for $b_m = \left(m - \frac32\right) / \alpha$,
and in the last step (Equation~\ref{eqn:inner-sum-bound w/ alpha and minus term - step2}), we use that $\alpha\cdot \ln\left(m - \frac{3}{2}\right) + \frac{1}{b_0} \ge \alpha\cdot \ln(m)$,
which holds for any $m\ge b_0 + 1$, $\alpha \le 2/3$, and $b_0\ge 1$.
\end{proof}

We now apply Claim~\ref{clm:inner-sum-bound w/ alpha and minus
  term} to~(\ref{eqn:Lch-sum-bound1_2 w/ alpha and minus term}) to
obtain

\begin{align}
\sum_{f = 1}^{\hat{e}_q} &\left(\sum_{j = 0}^{b_0 + f-1} \left(\frac{\alpha\cdot (b_{j+1} - b_j) + 1}{b_j} \right)
		- \frac{1}{2b_{b_0 + f - 1}}\right)
\\
&\ge \sum_{f=1}^{\hat{e}_q}\alpha\left( 1+\ln\left(\frac1\alpha\right)-\ln(b_0)+\ln(b_0+f)\right)
\nonumber\\
&= \alpha\cdot \hat{e}_q\cdot  \left(1+\ln\left(\frac1\alpha\right)-\ln(b_0)\right)
+ \alpha \sum_{f=1}^{\hat{e}_q}\ln(b_0+f)
\nonumber\\
&= \alpha\cdot \hat{e}_q\cdot  \left(1+\ln\left(\frac1\alpha\right)-\ln(b_0)\right)
+ \alpha\cdot \ln\left(\prod_{f=1}^{\hat{e}_q}(b_0+f)\right)
\nonumber\\
&= \alpha\cdot \hat{e}_q\cdot  \left(1+\ln\left(\frac1\alpha\right)-\ln(b_0)\right)
+ \alpha\cdot \ln\left(\frac{(\hat{e}_q+b_0)!}{b_0!}\right)
\nonumber\\
&> \alpha\cdot \hat{e}_q\cdot  \left(1+\ln\left(\frac1\alpha\right)-\ln(b_0)\right)
+ \alpha\cdot\left((\hat{e}_q+b_0)\cdot \ln\left(1 + \frac{\hat{e}_q}{b_0}\right) + \hat{e}_q\cdot (\ln(b_0) - 1)\right)
\label{eqn:L-increases sum final calc - strict ineq}\\
&= \alpha\cdot\hat{e}_q\cdot \left(1 + \frac{b_0}{\hat{e}_q}\right)\cdot \ln\left(1 + \frac{\hat{e}_q}{b_0}\right)+ \alpha\cdot\hat{e}_q\cdot \ln\left(\frac1\alpha\right)\,,
\end{align} 
where the strict inequality~\eqref{eqn:L-increases sum final calc - strict ineq} follows from Stirling's approximation for factorials (see Fact~\ref{fact:log-factorial} below).
This concludes the proof of Lemma~\ref{lem:L-increases sum}.
\end{proof}

Summing up the lower bound on the total S-increase with the lower bound on the total L-increase from Lemma~\ref{lem:L-increases sum},
we obtain the following lower bound (where $b_0 = \left\lceil\sigma^{t_q}_q\right\rceil - 1$):
\begin{equation}\label{eqn:total_gain_for_q}
\sum_i \Delta^i \Phi_q \ge (1 - \alpha - \beta)\cdot b_0
+ \alpha\cdot\hat{e}_q\cdot \left(1 + \frac{b_0}{\hat{e}_q}\right)\cdot \ln\left(1 + \frac{\hat{e}_q}{b_0}\right)+ \alpha\cdot\hat{e}_q\cdot \ln\left(\frac1\alpha\right)\,.
\end{equation}

\subsection{An Application of Stirling's Approximation}
\label{app:sum}

\begin{fact}\label{fact:log-factorial}
For any two positive integers $n > m$, we have
\[
\ln\left(\frac{n!}{m!}\right) >  n\ln\left(\frac{n}{m}\right) + (n-m)(\ln(m)-1) \enspace .
\]
\end{fact}
\begin{proof}
By Stirling's approximation, we have that for any positive integer $n$,
\[
\ln (\sqrt{2\pi n}) + n(\ln(n)-1)<\ln(n!)<\ln (\sqrt{2\pi n}) + n(\ln(n)-1)+\frac{1}{12n} \enspace.
\]
This gives
\begin{align*}
\ln\left(\frac{n!}{m!}\right) &> \ln (\sqrt{2\pi n}) + n(\ln(n)-1) - \ln (\sqrt{2\pi m}) - m(\ln(m)-1)-\frac{1}{12m}\\
 &= \left(n+\frac12\right)\ln\left(\frac{n}{m}\right) + (n-m)(\ln(m)-1)-\frac{1}{12m}\\
& \ge n\ln\left(\frac{n}{m}\right) + (n-m)(\ln(m)-1)+\frac12\cdot\ln\left(1+\frac{1}{m}\right)-\frac{1}{12m}\\
& \ge n\ln\left(\frac{n}{m}\right) + (n-m)(\ln(m)-1)\,,
\end{align*}
where the third step uses $n > m$ and the last step holds for any $m\ge 1$.
\end{proof}

\subsection{Calculation of the Competitive Ratio Upper Bound} %
\label{sec:ratio}

According to the following lemma, we can have $\rho = 1.4455154$ in~\eqref{eqn:mainIneq},
which implies that the competitive ratio of \LQD is at most $1 + 1 / \rho < 1.6918$,
according to the discussion in Section~\ref{s:setup of analysis}.
Thus, the following lemma concludes the proof of Theorem~\ref{thm:main}.

\begin{lemma}\label{lem:gainForEachQueue}
Consider a queue $q$ with $\hat{e}_q > 0$.
For any values of $\sigma^{t_q}_q$ and $\hat{e}_q$, 
it holds that $\sum_i \Delta^i \Phi_q \ge \rho\cdot \hat{e}_q$
for $\rho = 1.4455154$. 
\end{lemma}

\begin{proof}
As before, let $b_0 = \left\lceil\sigma^{t_q}_q\right\rceil - 1$.
Using inequality~\eqref{eqn:total_gain_for_q} as a lower bound on
$\sum_i \Delta^i \Phi_q$, it is sufficient to show
\[ (1 - \alpha - \beta)\cdot b_0
+ \alpha\cdot\hat{e}_q\cdot \left(1 + \frac{b_0}{\hat{e}_q}\right)\cdot \ln\left(1 + \frac{\hat{e}_q}{b_0}\right)+ \alpha\cdot\hat{e}_q\cdot \ln\left(\frac1\alpha\right)\ge \rho\cdot \hat{e}_q\,.
\]
Dividing by $\hat{e}_q$ and defining $x=\frac{\hat{e}_q}{b_0}$ gives
us that the optimal choice of $\rho$ satisfies,
\begin{align}\label{eqn:rho}
  \rho := \sup_{0 < \alpha \le 2/3}\,\, \inf_{x>0}\left((1 - \alpha - \beta)\cdot \frac{1}{x}
+ \alpha\cdot\left(1 + \frac{1}{x}\right)\cdot \ln\left(1 + x\right)+ \alpha\cdot\ln\left(\frac1\alpha\right)\right)\,.
\end{align}
Here, we also take into account that we require $\alpha \le 2/3$ in the analysis.
Recall also that $\beta=1-\sqrt{1-\alpha}-\alpha/2$.

For a fixed $\alpha$, the partial derivative of the RHS w.r.t.\ $x$ is 
$\alpha\cdot\left(\frac{1}{x}-\frac{\ln(1+x)}{x^2}\right)-\frac{1-\alpha-\beta}{x^2}$.
By setting the derivative to zero, we have:
\begin{align}\label{eqn:rho partial derivative}
  \frac{1-\alpha-\beta}{\alpha} + \ln(1+x) = x\,.
\end{align}
This is equivalent to $x = \exp\left(\frac{\alpha+\beta-1}{\alpha} + x\right) -1$,
which in turn has solution $
  x = -1-W_{-1}\left(-\exp\left(\frac{\beta-1}{\alpha}\right)\right) %
$; %
this is the value of $x$ that minimizes the expression for $\rho$ in~\eqref{eqn:rho}.
Here, $W_{-1}$ is the lower branch of the Lambert $W$ function on real numbers in $[1/e, 0)$; this function is also called product logarithm and is
defined for any integer $k$ by $y\cdot \exp(y) = z$ if and only if $y = W_k(z)$ (that is, $W_k$ is an inverse of $f(y) = y\cdot \exp(y)$).

Using also~\eqref{eqn:rho partial derivative} to replace $\ln\left(1 + x\right)$ by $x - \frac{1-\alpha-\beta}{\alpha}$ in~\eqref{eqn:rho} and substituting  $\beta=1-\sqrt{1-\alpha}-\alpha/2$, the optimal choice of $\rho$ satisfies
\begin{align*}
  \rho := \sup_{0 < \alpha \le 2/3}\Bigg(-\alpha\cdot W_{-1}\left(-\exp\left(-\frac{\sqrt{1-\alpha}}{\alpha}-\frac{1}{2}\right)\right)
  -\sqrt{1-\alpha}+\frac{\alpha}{2}
+ \alpha\cdot\ln\left(\frac1\alpha\right)\Bigg)\,.
\end{align*}
Optimizing through mathematical software, we get that the optimal choice for $\alpha$ is approximately $0.618906$
for which $\rho\ge 1.4455154$ and therefore, the competitive ratio of \LQD is at most $1 + 1/\rho \le 1.6917948$.
\end{proof}

\bibliography{references}

\begin{thebibliography}{10}

\bibitem{AielloKM08}
W.~Aiello, A.~Kesselman, and Y.~Mansour.
\newblock Competitive buffer management for shared-memory switches.
\newblock {\em {ACM} Transactions on Algorithms}, 5(1):3:1--3:16, 2008.

\bibitem{AEW18a}
K.~Al{-}Bawani, M.~Englert, and M.~Westermann.
\newblock Online packet scheduling for {CIOQ} and buffered crossbar switches.
\newblock {\em Algorithmica}, 80(12):3861--3888, 2018.

\bibitem{AlbersS05}
S.~Albers and M.~Schmidt.
\newblock On the performance of greedy algorithms in packet buffering.
\newblock {\em {SIAM} Journal on Computing}, 35(2):278--304, 2005.

\bibitem{andelman_queueing_policies_03}
N.~Andelman, Y.~Mansour, and A.~Zhu.
\newblock Competitive queueing policies for {QoS} switches.
\newblock In {\em Proceedings of the 14th ACM-SIAM Symposium on Discrete
  Algorithms (SODA)}, pages 761--770, 2003.

\bibitem{AzarL06}
Y.~Azar and A.~Litichevskey.
\newblock Maximizing throughput in multi-queue switches.
\newblock {\em Algorithmica}, 45(1):69--90, 2006.

\bibitem{AzarR05}
Y.~Azar and Y.~Richter.
\newblock Management of multi-queue switches in {QoS} networks.
\newblock {\em Algorithmica}, 43(1-2):81--96, 2005.

\bibitem{BFKMSS04}
N.~Bansal, L.~Fleischer, T.~Kimbrel, M.~Mahdian, B.~Schieber, and
  M.~Sviridenko.
\newblock Further improvements in competitive guarantees for {QoS} buffering.
\newblock In {\em Proceedings of the 31st International Colloquium on Automata,
  Languages and Programming (ICALP)}, pages 196--207, 2004.

\bibitem{bienkowski_randomized_algorithms_11}
M.~Bienkowski, M.~Chrobak, and Ł~Jeż.
\newblock Randomized competitive algorithms for online buffer management in the
  adaptive adversary model.
\newblock {\em Theoretical Computer Science}, 412(39):5121--5131, 2011.

\bibitem{bochkov2019new}
I.~Bochkov, A.~Davydow, N.~Gaevoy, and S.~I. Nikolenko.
\newblock New competitiveness bounds for the shared memory switch.
\newblock {\em CoRR}, abs/1907.04399, 2019.

\bibitem{bruno1998early}
J.~L. Bruno, B.~{\"{O}}zden, A.~Silberschatz, and H.~Saran.
\newblock Early fair drop: a new buffer management policy.
\newblock {\em Multimedia Computing and Networking}, 3654:148--161, 1998.

\bibitem{chamberland2000overall}
S.~Chamberland and B.~Sans{\`{o}}.
\newblock Overall design of reliable {IP} networks with performance guarantees.
\newblock In {\em Proceedings of the IEEE International Conference on
  Communications: Global Convergence Through Communications (ICC)}, pages
  1145--1151, 2000.

\bibitem{ChaoGuoBook}
H.~J. Chao and X.~Guo.
\newblock {\em Quality of Service Control in High-Speed Networks}.
\newblock Wiley-IEEE Press, 2001.

\bibitem{ChaoLiuBook}
H.~J. Chao and B.~Liu.
\newblock {\em High Performance Switches and Routers}.
\newblock Wiley-IEEE Press, 2007.

\bibitem{CCFJST06}
F.~Y.~L. Chin, M.~Chrobak, S.~P.~Y. Fung, W.~Jawor, J.~Sgall, and
  T.~Tich{\'{y}}.
\newblock Online competitive algorithms for maximizing weighted throughput of
  unit jobs.
\newblock {\em Journal of Discrete Algorithms}, 4(2):255--276, 2006.

\bibitem{chin_partial_job_values_03}
F.~Y.~L. Chin and S.~P.~Y. Fung.
\newblock Online scheduling with partial job values: Does timesharing or
  randomization help?
\newblock {\em Algorithmica}, 37(3):149--164, 2003.

\bibitem{CJST07}
M.~Chrobak, W~Jawor, J.~Sgall, and T.~Tich{\'{y}}.
\newblock Improved online algorithms for buffer management in {QoS} switches.
\newblock {\em {ACM} Transactions on Algorithms}, 3(4):50, 2007.

\bibitem{EnglertW09}
M.~Englert and M.~Westermann.
\newblock Lower and upper bounds on {FIFO} buffer management in {QoS} switches.
\newblock {\em Algorithmica}, 53(4):523--548, 2009.

\bibitem{EW12}
M.~Englert and M.~Westermann.
\newblock Considering suppressed packets improves buffer management in quality
  of service switches.
\newblock {\em {SIAM} Journal on Computing}, 41(5):1166--1192, 2012.

\bibitem{EKNS14}
P.~Eugster, K.~Kogan, S.~Nikolenko, and A.~Sirotkin.
\newblock Shared memory buffer management for heterogeneous packet processing.
\newblock In {\em Proceedings of the 34th IEEE International Conference on
  Distributed Computing Systems (ICDCS)}, pages 471--480, 2014.

\bibitem{Goldwasser10}
M.~H. Goldwasser.
\newblock A survey of buffer management policies for packet switches.
\newblock {\em {SIGACT} News}, 41(1):100--128, 2010.

\bibitem{HahneKM01}
E.~L. Hahne, A.~Kesselman, and Y.~Mansour.
\newblock Competitive buffer management for shared-memory switches.
\newblock In {\em Proceedings of the 13th ACM Symposium on Parallelism in
  Algorithms and Architectures (SPAA)}, pages 53--58, 2001.

\bibitem{hajek_unit_packets_01}
B.~Hajek.
\newblock On the competitiveness of on-line scheduling of unit-length packets
  with hard deadlines in slotted time.
\newblock In {\em Proceedings of the 35th Conference on Information Sciences
  and Systems}, pages 434--438, 2001.

\bibitem{Jez13}
Ł. Jeż.
\newblock A universal randomized packet scheduling algorithm.
\newblock {\em Algorithmica}, 67(4):498--515, 2013.

\bibitem{KesselmanLMPSS04}
A.~Kesselman, Z.~Lotker, Y.~Mansour, B.~Patt{-}Shamir, B.~Schieber, and
  M.~Sviridenko.
\newblock Buffer overflow management in {QoS} switches.
\newblock {\em {SIAM} Journal on Computing}, 33(3):563--583, 2004.

\bibitem{KesselmanM04}
A.~Kesselman and Y.~Mansour.
\newblock Harmonic buffer management policy for shared memory switches.
\newblock {\em Theoretical Computer Science}, 324(2-3):161--182, 2004.

\bibitem{KesselmanMS05}
A.~Kesselman, Y.~Mansour, and R.~van Stee.
\newblock Improved competitive guarantees for {QoS} buffering.
\newblock {\em Algorithmica}, 43(1-2):63--80, 2005.

\bibitem{KesselmanR06}
A.~Kesselman and A.~Ros{\'{e}}n.
\newblock Scheduling policies for {CIOQ} switches.
\newblock {\em Journal of Algorithms}, 60(1):60--83, 2006.

\bibitem{KobayashiMO07}
K.~M. Kobayashi, S.~Miyazaki, and Y.~Okabe.
\newblock A tight bound on online buffer management for two-port shared-memory
  switches.
\newblock In {\em Proceedings of the 19th ACM Symposium on Parallelism in
  Algorithms and Architectures (SPAA)}, pages 358--364, 2007.

\bibitem{LSS07}
F.~Li, J.~Sethuraman, and C.~Stein.
\newblock Better online buffer management.
\newblock In {\em Proceedings of the 18th ACM-SIAM Symposium on Discrete
  Algorithms (SODA)}, pages 199--208, 2007.

\bibitem{Matsakis15}
N.~Matsakis.
\newblock {\em Approximation Algorithms for Packing and Buffering problems}.
\newblock PhD thesis, University of Warwick, {UK}, 2015.

\bibitem{nabeshima2005performance}
M.~Nabeshima and K.~Yata.
\newblock Performance improvement of active queue management with per-flow
  scheduling.
\newblock {\em IEE Proceedings-Communications}, 152(6):797--803, 2005.

\bibitem{NikolenkoK16}
S.~I. Nikolenko and K.~Kogan.
\newblock Single and multiple buffer processing.
\newblock In {\em Encyclopedia of Algorithms}, pages 1988--1994. Springer,
  2016.

\bibitem{suter1998design}
B.~Suter, T.~V. Lakshman, D.~Stiliadis, and A.~K. Choudhury.
\newblock Design considerations for supporting {TCP} with per-flow queueing.
\newblock In {\em Proceedings of the 17th IEEE Conference on Computer
  Communications (INFOCOM)}, pages 299--306, 1998.

\bibitem{VeselyCJS22}
Pavel Vesel{\'{y}}, Marek Chrobak, Lukasz Jez, and Jir{\'{\i}} Sgall.
\newblock A $\phi$-competitive algorithm for scheduling packets with deadlines.
\newblock {\em {SIAM} J. Comput.}, 51(5):1626--1691, 2022.

\bibitem{WCH}
S.~X. Wei, E.~J. Coyle, and M.~T. Hsiao.
\newblock An optimal buffer management policy for high-performance packet
  switching.
\newblock In {\em Proceedings of the Global Communication Conference
  (GLOBECOM)}, pages 924--928, 1991.

\bibitem{Zhu04}
A.~Zhu.
\newblock Analysis of queueing policies in {QoS} switches.
\newblock {\em Journal of Algorithms}, 53(2):137--168, 2004.

\end{thebibliography}

\newpage
\appendix

\section{Glossary and Notation}
\label{sec:notations}

Table~\ref{tab:glossary} provides a list of concepts and notation used throughout the paper
(apart from those that are used locally, for example, inside a proof of a single lemma).
The majority of them are defined in Section~\ref{s:splitting}, apart from the first eight from Section~\ref{s:setup of analysis} and 
the last eight from Sections~\ref{s:liveAndDying} and~\ref{sec:puttingItAllTogether}.

\begin{table}[ht]
	\centering
	\caption{\small List of concepts and notation used throughout the paper }
	\label{tab:glossary}
	{\small
	\begin{tabular}{c|p{13.5cm}}
		$M$ & buffer size \\ \hline
		\OPT-extra packet & a packet transmitted by \OPT from a queue $q$ at some time $t$ such that \LQD transmits no packet from $q$ in step $t$ \\ \hline
		\LQD-extra packet & a packet transmitted by \LQD from a queue $q$ at some time $t$ such that \OPT transmits no packet from $q$ in step $t$ \\ \hline
		$\EXTRA$ & the total number of transmitted \OPT-extra packets \\ \hline
        $\LEXTRA$ & the total number of transmitted \LQD-extra packets \\ \hline
        $\rho$ & a parameter in~\eqref{eqn:mainIneq} determining the competitive ratio upper bound \\ \hline
        $e_q$ & the total number of transmitted \OPT-extra packets from queue $q$ \\ \hline
        $m_q$ & \# of transmitted \LQD-extra packets which are mapped to $q$, satisfying $m_q \le e_q$ (def.\ in Sec.~\ref{sec:mapping}) \\ \hline
        $\hat{e}_q = e_q - m_q$ & the total number of transmitted \OPT-extra packets from $q$ that are not canceled out \\ \hline
        $s^t_{\OPT}(q)$ & \# of packets in queue $q$ in the \OPT buffer in step $t$ \\ \hline
		$s^t_{\LQD}(q)$ & \# of packets in queue $q$ in the \LQD{} buffer in step $t$ \\ \hline
		$s^t_{\max} = \max_q s^t_{\LQD}(q)$ & the maximal size of a queue in the \LQD buffer in step $t$ \\ \hline
		active/inactive & a queue $q$ is active at $t$ if $s^t_{\OPT}(q)\ge 1$ or $s^t_{\LQD}(q)\ge 1$, and inactive otherwise \\ \hline
		overflow & a queue $q$ overflows in step $t$ if (i) a packet destined to $q$ is evicted by \LQD at $t$, or 
(ii) the \LQD buffer is full in step $t$, $s^t_{\LQD}(q)\ge s^t_{\max} - 1$, and $s^t_{\LQD}(q) \ge 1$, or both \\ \hline
		$t_q$ & the last time step in which queue $q$ overflows; if $q$ does not overflow in any step, we define $t_q = -1$ \\ \hline
		$e^t_q$ & \# of \OPT-extra packets transmitted from $q$ in step $t$ or later \\ \hline
		$\hat{e}^t_q = \max\{e^t_q-m_q,0\}$ & $e^t_q$ adjusted for the packets that are canceled out by transmitted \LQD-extra packets \\ \hline
		$\tau_{1}<\tau_{2}<...<\tau_{\ell}$ & the time steps in which at least one queue overflows for the last time, i.e., for each $1\leq i\leq \ell$ there is a queue $q$ such that $\tau_{i}=t_{q}\ge 0$ \\ \hline
		phase $i$ & time interval $[\tau_i, \tau_{i+1})$ (if $i=\ell$, then let $\tau_{i+1} = \infty$) \\ \hline
		$\Phi_q$ & counter for keeping track of the \LQD profit assigned to queue $q$ \\ \hline 
		$\Delta^i \Phi_q$ & the profit assigned to queue $q$ in phase $i$ \\ \hline
		$\alpha$ & a parameter for splitting the \LQD profit (set to $\approx 0.619$ at the very end) \\ \hline
		$\Psi^i$ & a potential for amortizing the \LQD profit, equal to $\alpha\cdot$ (\# of active queues overflowing after $\tau_i$) \\ \hline
		$u^i$ & \# of queues active in step $\tau_{i+1}$ that were inactive in step $\tau_i$ and will overflow after $\tau_{i+1}$ \\ \hline %
		$v^i$ & \# of queues $q$ with $t_q = \tau_{i+1}$ that are active in step $\tau_i$ \\ \hline
		$\Delta^i \Psi = \Psi^{i+1} - \Psi^i$ & the change of the potential in phase $i$, equal to $\alpha\cdot (u^i - v^i)$ \\ \hline
		$o^i$ & \# of packets that \LQD transmits in phase $i$ from queues $q$ with $\tau_i\ge t_q$ and $e_q > 0$ \\ \hline
		$n^i$ & \# of packets that \LQD transmits in phase $i$ from queues $q$ with $\tau_i < t_q$ or $e_q = 0$ \\ \hline
		$o^i_q$ & \# of packets that \LQD transmits in phase $i$ from a given queue $q$ with $\tau_i \ge t_q$ or $e_q > 0$ \\ \hline
		$j_q = \min\{j : \hat{e}^j_q = 0\}$ & the index $j$ of the earliest step $\tau_j$ in which all remaining \OPT-extra packets to be transmitted from $q$ (if any) are canceled out \\ \hline
		Live/dying  & sets of queues defined w.r.t.\ a queue $q$ with $t \ge t_q$ and $\hat{e}_q > 0$; see Def.~\ref{def:liveDying} in Sec.~\ref{s:liveAndDying} \\ \hline
		$\calL^i_q$ & the set of queues that are live in step $\tau_i$ w.r.t.\ queue $q$ \\ \hline
		$\calD^i_q$ & the set of queues that are dying in step $\tau_i$ w.r.t.\ queue $q$ \\ \hline
		$d^i_q$ & \# of packets transmitted from queues in $\calD^i_q$ during phase $i$ \\ \hline
		$\sigma^i_q$ & the average number of \LQD packets in live queues (w.r.t.\ $q$) in step $\tau_i$ \\ \hline
		L-increase & the first part of the \LQD profit assigned to a queue in a phase, see~\eqref{eqn:splitting profit def} \\ \hline
		S-increase & the second part of the \LQD profit assigned to a queue in a phase, see~\eqref{eqn:splitting profit def} \\ \hline
		$\beta$ & another parameter for splitting the \LQD profit, equal to $1-\sqrt{1-\alpha}-\alpha/2 \approx 0.0732$ \\ \hline
	\end{tabular}
	\vspace{-10em}
}
\end{table}
\end{document}